\documentclass[a4paper,12pt]{article}
\usepackage{amsmath,amssymb,amsthm}

\newtheorem{lemma}{Lemma}
\newtheorem{theorem}{Theorem}

\theoremstyle{definition}

\theoremstyle{remark}

\newcommand{\beq}{\begin{eqnarray}}
\newcommand{\eeq}{\end{eqnarray}}
\newcommand{\beqnn}{\begin{eqnarray*}}
\newcommand{\eeqnn}{\end{eqnarray*}}
\newcommand{\rd}{\partial}

\newcommand{\Tr}{\operatorname{Tr}}
\newcommand{\diag}{\operatorname{diag}}
\newcommand{\tp}[1]{\:{}^{\mathrm{t}}#1}
\newcommand{\CC}{\mathbf{C}}
\newcommand{\PP}{\mathbf{P}}

\newcommand{\ZZ}{\mathbf{Z}}
\newcommand{\bst}{\boldsymbol{t}}
\newcommand{\bsx}{\boldsymbol{x}}

\newcommand{\calB}{\mathcal{B}}
\newcommand{\calF}{\mathcal{F}}
\newcommand{\calH}{\mathcal{H}}
\newcommand{\calL}{\mathcal{L}}
\newcommand{\calM}{\mathcal{M}}
\newcommand{\calP}{\mathcal{P}}
\newcommand{\calQ}{\mathcal{Q}}
\newcommand{\gl}{\mathrm{gl}}
\newcommand{\GL}{\mathrm{GL}}
\newcommand{\res}{\mathop{\mathrm{res}}}

\begin{document}

\title{Generalized string equations for\\
double Hurwitz numbers}
\author{Kanehisa Takasaki%
\thanks{E-mail: takasaki@math.h.kyoto-u.ac.jp}\\
{\small
Graduate School of Human and Environmental Studies,
Kyoto University}\\
{\small Yoshida, Sakyo, Kyoto, 606-8501, Japan}}
\date{}
\maketitle

\begin{abstract}
The generating function of double Hurwitz numbers 
is known to become a tau function of the Toda hierarchy. 
The associated Lax and Orlov-Schulman operators turn out 
to satisfy a set of generalized string equations.  
These generalized string equations resemble those of 
$c = 1$ string theory except that the Orlov-Schulman operators 
are contained therein in an exponentiated form.  
These equations are derived from a set of intertwining 
relations for fermion bilinears in a two-dimensional 
free fermion system.  The intertwiner is constructed 
from a fermionic counterpart of the cut-and-join operator.  
A classical limit of these generalized string equations 
is also obtained.  The so called Lambert curve emerges 
in a specialization of its solution.  This seems 
to be another way to derive the spectral curve of 
the random matrix approach to Hurwitz numbers.  
\end{abstract}
\bigskip

\begin{flushleft}
\end{flushleft}

\section{Introduction}

Hurwitz numbers count the topological types of 
finite ramified coverings of a given Riemann surface 
\cite{Hurwitz1891}.  Some ten years ago, 
Hurwitz numbers of coverings of the Riemann sphere $\CC\PP^1$ 
turned out to be related to Hodge integrals on 
the Deligne-Mumford moduli space $\bar{\calM}_{g,n}$ 
of marked stable curves \cite{FP99,GJV99,ELSV00}
and Gromov-Witten invariants of $\CC\PP^1$ 
\cite{Pandharipande00,Okounkov00,OP01,OP02}.  
These observations led to various developments 
and applications, such as new proofs \cite{KL05,KL06,CLL06} 
of Witten's conjecture \cite{Witten91} 
(Kontsevich's theorem \cite{Kontsevich92}) 
on two-dimensional topological gravity, in other words, 
integrals of $\psi$-classes on $\bar{\calM}_{g,n}$.  
Recently, a set of new recursion relations 
for Hurwitz numbers were derived \cite{BM09,BEMS09,EMS09}
as an analogue of Eynard and Orantin's 
``topological recursion relations'' \cite{EO07}.  

These developments in the last decade are also 
more or less connected with integrable hierarchies 
of the KdV, KP and Toda type (see reviews 
by the Kyoto group \cite{DJKM83,MJD-book,TT95} 
for basic knowledge on those integrable systems).  
Firstly, several Toda-like equations show up 
in relation with Gromov-Witten invariants of $\CC\PP^1$.  
Secondly, the Witten conjecture is formulated 
in the language of the KdV hierarchy and 
its Virasoro symmetries.  Subsequent studies 
\cite{GJ08,MM08,Kazarian08} 
on the Witten conjecture from the Hurwitz side 
revealed a connection with the KP hierarchy as well.  

A master integrable system in this sense 
is the (two-dimensional) Toda hierarchy \cite{UT84}.  
As pointed out by Okounkov \cite{Okounkov00}, 
a generating function of ``double Hurwitz numbers'' 
gives a special solution of the Toda hierarchy.  
Being a tau function of the Toda hierarchy, 
this generating function has two sets of 
independent variables $\bst = (t_1,t_2,\cdots)$ 
and $\bar{\bst} = (\bar{t}_1,\bar{t}_2,\cdots)$. 
By specializing part of these variables to 
particular values, the generating function 
of ``simple Hurwitz numbers'' are recovered.  
The aforementioned Toda-like equations and KP hierarchy 
may be thought of as equations satisfied 
by those specialized tau functions.  

The goal of this paper is to derive 
the {\it generalized string equations}
\beqnn
  L = Qe^{-\beta/2}\bar{L}e^{\beta\bar{M}},\quad 
  \bar{L}^{-1} = Qe^{\beta/2}L^{-1}e^{\beta M}
\eeqnn
for this special solution of the Toda hierarchy, 
and to examine implications thereof.  
Here $\beta$ and $Q$ are parameters 
prepared along with $\bst$ and $\bar{\bst}$. 
$L,\bar{L}$ and $M,\bar{M}$ are, respectively, 
the Lax and Orlov-Schulman operators, all being 
one-dimensional difference operators. 
Compared with the cases of two-dimensional quantum gravity 
\cite{AvM92} and $c = 1$ string theory 
\cite{DMP93,HOP94,EK94,Takasaki95,NTT95,Takasaki96}, 
these generalized string equations have a notable new feature. 
Namely, they contain the exponentials $e^{\beta M}$ 
and $e^{\beta\bar{M}}$ of the Orlov-Schulman operators.  
This is in sharp contrast with the conventional 
string equations that are linear in $M$ and $\bar{M}$ 
(except for the the case of a deformed version 
of $c = 1$ string theory \cite{NTT95}, in which 
the string equations are quadratic in $M$ and $\bar{M}$).  

The exponentials of $M$ and $\bar{M}$ are related 
to the ``quantum torus algebra'' in a two-dimensional 
free fermion system.  The same Lie algebraic structure  
plays an important role in the melting crystal model 
of five-dimensional gauge theory as well \cite{NT07,NT08}.  
We borrow some technical ideas developed therein. 
The generalized string equations are derived 
from algebraic relations (referred to as 
``intertwining relations'') among fermion bilinears.  
A clue therein is a special fermion bilinear that corresponds 
to the so called ``cut-and-join operator'' \cite{GJ97}.  

The generalized string equations also have 
a {\it classical limit} of the form 
\beqnn
  \calL = Q\bar{\calL}e^{\beta\bar{\calM}},\quad
  \bar{\calL}^{-1} = Q\calL^{-1}e^{\beta\calM}, 
\eeqnn
where $\calL,\calM,\bar{\calL},\bar{\calM}$ are 
``long-wave'' or ``dispersionless'' limits \cite{TT95} 
of the Lax and Orlov-Schulman operators.  
Conceptually, the classical limit amounts to 
the genus-zero part of Hurwitz numbers. 
Similarity with $c = 1$ string theory becomes 
even more obvious in this limit (which amounts 
to the genus-0 part of string amplitudes). 
The generalized string equations of $c = 1$ string theory 
\cite{DMP93,HOP94,EK94,Takasaki95} 
have a classical limit of the form 
\beqnn 
  \calL = \bar{\calL}\bar{\calM},\quad 
  \bar{\calL}^{-1} = \calL^{-1}\calM. 
\eeqnn
Thus, roughly speaking, the linear terms $\calM,\bar{\calM}$ 
therein are now replaced by the exponential terms 
$e^{\beta\calM},e^{\beta\bar{\calM}}$. 
Thanks to this similarity, one can apply the method 
developed for solving the generalized string equations 
of $c = 1$ string theory \cite{Takasaki95} 
to construct a solution of the foregoing equations 
in the form of power series of $t_k$'s and $\bar{t}_k$'s.  
This construction simplifies when either $t_k$'s 
or $\bar{t}_k$'s are specialized to particular values 
for which the tau function reduces to the generating function 
of simple Hurwitz numbers.  Remarkably, 
we encounter therein the equation 
\beqnn
  x = ye^y
\eeqnn
of the ``Lambert curve'' that lies in the heart 
of the new recursion relations \cite{BM09,BEMS09,EMS09}.  
This should not be a coincidence.  Presumably, 
the classical limit of the generalized string equations 
will be another way to derive the ``spectral curve'' 
in the random matrix approach \cite{BEMS09,MS09a,MS09b}. 

This paper is organized as follows.  Section 2 reviews 
the notion of Hurwitz numbers and their generating functions.  
Schur functions and their special values are used 
to interpret the generating functions as tau functions. 
Section 3 prepares technical tools from a two-dimensional 
free fermion system.  We introduce relevant fermion bilinears, 
and recall Okounkov's result on a fermionic representation 
of the generating function of double Hurwitz numbers \cite{Okounkov00}. 
Section 4 presents the generalized string equations. 
We start from intertwining relations of fermion bilinears, 
and translate those relations to the generalized string equations.  
Section 5 formulates the classical limit of the generalized 
string equations.   The classical limit is first derived 
by heuristic consideration, then justified by showing 
that the associated tau function has an $\hbar$-expansion.  
Section 6 is devoted to solving this classical limit 
of the generalized string equations.  The solution 
is constructed in much the same way as that of 
the generalized string equations for $c = 1$ string theory.  
The Lambert curve shows up in a specialization of this solution.

\section{Generating functions of Hurwitz numbers}

In this section, we use various notions and formulae 
on partitions, Young diagrams and Schur functions. 
They are mostly borrowed from Macdonald's book 
\cite{Macdonald-book}.

\subsection{Hurwitz numbers of Riemann sphere}

Let us consider finite ramified coverings of $\CC\PP^1$.  
Two coverings $\pi: \Gamma \to \CC\PP^1$ and 
$\pi': \Gamma' \to \CC\PP^1$ are said to be 
topologically equivalent if $\pi$ and $\pi's$ are 
connected by a homeomorphism $\phi: \Gamma \to \Gamma'$ 
as $\pi = \pi'\circ\phi$.  Let $[\pi]$ denote 
the equivalence class of the covering $\pi$.  
Note that $\Gamma$ can be a disconnected surface.  

Given a positive integer $d$ and a set of 
distinct points $P_1,\cdots,P_r$ of $\CC\PP^1$, 
there are only a finite number of nonequivalent 
$d$-fold coverings that are unramified 
over all points other than $P_1,\cdots,P_r$.   
Those nonequivalent coverings can be further classified 
by the ramification data of sheets over $P_1,\cdots,P_r$.  
These data are given by conjugacy classes $C_1,\cdots,C_r$ 
of the $d$-th symmetric group $S_d$.  

A conjugacy class $C$ of $S_d$ is determined by 
the cycle type 
\beqnn
  \mu = (\mu_1,\cdots,\mu_l), \quad 
  \mu_1 \ge \cdots \ge \mu_l, \quad 
  |\mu| = \mu_1 + \cdots + \mu_l = d 
\eeqnn
of a representative $\sigma \in S_d$ of $C$, e.g., 
\beqnn
  \sigma = 
  (1,\cdots,\mu_1)(\mu_1+1,\cdots,\mu_1+\mu_2) 
  \cdots (\mu_1+\cdots+\mu_{l-1},\cdots,d), 
\eeqnn
where $(j_1,\cdots,j_m)$ denotes the cyclic permutation 
sending $j_1 \to j_2 \to\cdots \to j_m \to j_1$.  
The cycle type thus becomes a partition of $d$, 
and has another expression 
\beqnn
  \mu = (1^{m_1}2^{m_2}\cdots) 
\eeqnn
with the numbers $m_i$ of $i$-cycles.  
Let $C(\mu)$ denote the conjugacy class 
determined by a partition $\mu$ of $d$.  

A cyclic permutation $(j_1,\cdots,j_m)$ of length $m$ 
represents a cyclic covering of degree $m$ realized., 
e.g., by the Riemann sheets of $(z - a)^{1/m}$ 
above the point $a$.  We use the cycle type 
$\mu = (\mu_1,\cdots,\mu_l)$ as ramification data 
above a point of $\CC\PP^1$ to show that 
the sheets above that point locally look like 
a disjoint union of $l$ cyclic coverings 
of degree $\mu_1,\mu_2,\cdots,\mu_l$.  

The Hurwitz number $H_d(C_1,\cdots,C_r)$, 
also denoted by $H_d(\mu^{(1)},\cdots,\mu^{(r)})$ 
where $\mu^{(k)}$'s are the cycle types of $C_k$'s, 
is defined to be the sum 
\beqnn
  H_d(C_1,\cdots,C_r) 
  = \sum_{[\pi]}\frac{1}{|\mathrm{Aut}(\pi)|}
\eeqnn
of the weights $1/|\mathrm{Aut}(\pi)|$ over 
all equivalent classes $[\pi]$ of {\it possibly disconnected} 
coverings $\pi$ that are ramified over the $r$ points 
$P_1,\cdots,P_r$ with the ramification data $C_1,\cdots,C_r$.  
$\mathrm{Aut}(\pi)$ denotes the group of automorphisms of $\pi$.  
As the notation suggests, $H_d(C_1,\cdots,C_r)$ 
does not depend on the position of $P_1,\cdots,P_r$. 

The Hurwitz numbers can be determined by 
a genuinely group theoretical method 
(Burnside's theorem \cite{Burnside-book}).  
This leads to the beautiful formula 
\beq
  H_d(\mu^{(1)},\cdots,\mu^{(r)}) 
  = \sum_{|\lambda|=d} 
    \left(\frac{\dim\lambda}{d!}\right)^2
    \prod_{k=1}^r f_\lambda(\mu^{(k)}), 
\label{Burnside-formula}
\eeq
where the sum is over all partitions $\lambda$ of $d$.  
Let us explain the notations used in this formula.  

Firsly, $\dim\lambda$ denotes the dimension $\dim V_\lambda$ 
of the irreducible representation $(\rho_\lambda,V_\lambda)$ 
of $S_d$ determined by $\lambda$.  In other words, 
$\dim\lambda$ is the number of all standard tableaux 
of shape $\lambda$, and can be calculated by 
the so called hook length formula 
\beq
  \dim\lambda 
  = \dfrac{d!}{\displaystyle \prod_{(i,j)\in\lambda}h(i,j)}, 
\eeq
where $h(i,j)$ denotes the length of the hook cornered 
at the cell $(i,j) \in \lambda$.  Note that partitions 
are identified with the associated Young diagrams.  
As an immediate consequence, 
$\dim\lambda$ turns out to be symmetric 
under the transpose (or conjugate) 
$\lambda \mapsto \tp{\lambda}$ of the partition, namely, 
\beq
  \dim\tp{\lambda} = \dim\lambda. 
\label{dim(lambda)-symmetry}
\eeq

Secondly, $f_\lambda(\mu)$ denotes the value $f_\lambda(C(\mu))$ 
of the class function 
\beq
  f_\lambda(C) 
  = \frac{\chi_\lambda(C)}{\dim\lambda}|C| 
\eeq
on $S_d$ evaluated at $C = C(\mu)$.  
$\chi_\lambda$ denotes the character $\Tr_{V_\lambda}\rho_\lambda$. 
The cardinality $|C(\mu)|$ of $C(\mu) \subset S_d$ 
can be written as 
\beq
  |C(\mu)| = \frac{d!}{z_\mu},\quad 
  z_\mu = \prod_{i\ge 1}m_i!i^{m_i}. 
\eeq
The values of $f_\lambda(C)$ for the simplest two cases 
$C = (1^d),(1^{d-2}2)$ can be explicitly calculated as 
\beq
  f_\lambda(1^d) = 1, \quad 
  f_\lambda(1^{d-2}2) = \frac{\kappa_\lambda}{2}
\label{f(1^d),f(1^{d-2}2)}
\eeq
where 
\beqnn
  \kappa_\lambda 
  = \sum_{i=1}^l \lambda_i(\lambda_i-2i+1) 
  = \sum_{i=1}^l\left(\left(\lambda_i-i+\frac{1}{2}\right)^2 
      - \left(-i+\frac{1}{2}\right)^2 \right). 
\eeqnn
This number has another useful expression of the form 
\beq
  \kappa_\lambda = 2 \sum_{(i,j)\in\lambda} (i-j), 
\eeq
which implies, e.g., the anti-symmetric property  
\beq
  \kappa_{\tp{\lambda}} = - \kappa_\lambda. 
\label{kappa-antisymmetry}
\eeq

\subsection{Simple Hurwitz numbers}

Let us review the notion of simple Hurwitz numbers. 
They are the Hurwitz numbers such that the types 
of all but one ramification points $P_1,\cdots,P_r$ 
are restricted to $1^{d-2}2$.  The exceptional 
ramification point $P_{r+1}$ can have an arbitrary 
cycle type $\mu$.  The Hurwitz numbers of this type 
\beq
  H_d(\underbrace{1^{d-2}2,\cdots,1^{d-2}2}_r,\mu) 
  = \sum_{|\lambda|=d}
    \left(\frac{\dim\lambda}{d!}\right)^2 
    \left(\frac{\kappa_\lambda}{2}\right)^r
    f_\lambda(\mu) 
\label{simple-H_d}
\eeq
are called simple Hurwitz numbers. 
Introducing an extra (finite or infinite) set of 
variables $\bsx = (x_1,x_2,\cdots)$, one can construct 
a generating function of these numbers as 
\beqnn
  Z(\bsx) 
  = \sum_{r=0}^\infty\sum_{d=0}^\infty\sum_{|\mu|=d}
      H_d(\underbrace{1^{d-2}2,\cdots,1^{d-2}2}_r,\mu) 
      \frac{\beta^r}{r!}Q^dp_\mu, 
\eeqnn
where $p_\mu$'s are the monomials 
\beqnn
  p_\mu = p_{\mu_1}p_{\mu_2}\cdots 
\eeqnn
of the power sums 
\beqnn
  p_k = \sum_{i\ge 1}x_i^k, \quad k = 1,2,\cdots. 
\eeqnn
We now substitute (\ref{simple-H_d}) in the definition 
of $Z(\bsx)$ and use the Frobenius formula 
\beq
  \sum_{|\mu|=d}\frac{\chi_\lambda(C(\mu))}{z_\mu}p_\mu 
  = s_\lambda(\bsx) 
\label{Frobenius-formula}
\eeq
to rewrite the sum over $\mu$ into a Schur function 
times numerical factors.  The numerical factors 
partly cancel with some other factors.  
The generating function thus reduces to 
\beq
  Z(\bsx) 
  = \sum_{\lambda\in\calP}
    \frac{\dim\lambda}{|\lambda|!}
    e^{\beta\kappa_\lambda/2}Q^{|\lambda|}s_\lambda(\bsx). 
\eeq

$Z(\bsx)$ turns out to be a tau function of the KP hierarchy 
\cite{DJKM83,MJD-book}.  To see this, we change the variables 
from $\bsx$ to the standard time variables $\bst = (t_1,t_2,\cdots)$ 
of the KP hierarchy via the power sums as 
\beqnn
  t_k = \frac{p_k}{k} = \frac{1}{k}\sum_{i\ge 1}x_i^k 
\eeqnn
and consider the Schur functions $s_\lambda(\bsx)$ 
as functions of $\bst$.  It is convenient to use 
Zinn-Justin's notation $s_\lambda[\bst]$ for the latter  
\cite{Zinn-Justin09}.  

Actually, $s_\lambda[\bst]$ can be redefined directly.   
Let us recall the Jacobi-Trudi formula 
\beq
  s_\lambda(\bsx) = \det(h_{\lambda_i-i+j}(\bsx))_{i,j=1}^n, 
\eeq
where $\lambda_i$'s are parts of $\lambda$ 
(which are assumed to be equal to $0$ for $i > n$, namely, 
$\lambda = (\lambda_1,\cdots,\lambda_n,0,\cdots)$), 
and $h_m(\bsx)$'s are the complete symmetric functions 
defined by the generating function
\beqnn
  \sum_{m=0}^\infty h_m(\bsx)z^m 
  = \prod_{i\ge 1}(1 - x_iz)^{-1}. 
\eeqnn
Let $h_m[\bst]$ denotes $h_m(\bsx)$ considered as 
a function of $\bst$.  $h_m[\bst]$'s can be redefined 
by the generating function 
\beqnn
  \sum_{m=0}^\infty h_m[\bst]z^m 
  = \exp\left(\sum_{k=1}^\infty t_kz^k\right). 
\eeqnn
Consequently, $s_\lambda[\bst]$ can be expressed as 
\beq
  s_\lambda[\bst] = \det(h_{\lambda_i-i+j}[\bst])_{i,j=1}^n.
\eeq
One can thus replace $s_\lambda(\bsx)$ by $s_\lambda[\bst]$ 
to obtain the generating function 
\beq
  Z[\bst]  
  = \sum_{\lambda\in\calP}
    \frac{\dim\lambda}{|\lambda|!}
    e^{\beta\kappa_\lambda/2}Q^{|\lambda|}s_\lambda[\bst]. 
\label{Z[t]}
\eeq

Identifying $Z[\bst]$ as a KP tau function requires 
some more consideration.  Several methods are known 
in the literature \cite{KL06,GJ08,MM08,Kazarian08}.  
One can explain this fact in the context 
of the Toda hierarchy as well.  A clue is the formula 
\beq
  \frac{\dim\lambda}{|\lambda|!} 
  = s_\lambda[1,0,0,\cdots] 
\eeq
that can can be derived, e.g., 
from the Frobenius formula 
(\ref{Frobenius-formula}) by letting 
$p_1 = 1$ and $p_k = 0$ for $k > 1$.  
Consequently, $Z[\bst]$ can be rewritten as 
\beq
  Z[\bst] 
  = \sum_{\lambda\in\calP}
    e^{\beta\kappa_\lambda/2}Q^{|\lambda|}
    s_\lambda[\bst]s_\lambda[1,0,\cdots]. 
\eeq
This function is a specialization of 
the generating function $Z[\bst,\bar{\bst}]$ 
of double Hurwitz numbers introduced below.  
$Z[\bst,\bar{\bst}]$ is a tau function of the Toda hierarchy 
(or, rather, the two-component KP hierarchy\cite{DJKM83}, 
because the lattice coordinate of the Toda lattice 
is absent here).  It is well known \cite{UT84} 
that any Toda (or two-component KP) tau function 
is also a tau function of the KP hierarchy 
with respect to one of the two sets of variables.  
This implies that $Z[\bst]$ is a KP tau function.

\subsection{Double Hurwitz numbers} 

Let us choose yet another point $P_0$ of $\CC\PP^1$ 
of an arbitrary ramification type $\bar{\mu}$ 
in addition to the $r+1$ points in the case 
of simple Hurwitz numbers. The Hurwitz numbers of this type 
\beq
  H_d(\bar{\mu},\underbrace{1^{d-2}2,\cdots,1^{d-2}2}_r,\mu) 
  = \sum_{|\lambda|=d}
    \left(\frac{\dim\lambda}{d!}\right)^2 
    \left(\frac{\kappa_\lambda}{2}\right)^r
    f_\lambda(\mu) f_\lambda(\bar{\mu}) 
\label{double-H_d}
\eeq
are called double Hurwitz numbers.  
To construct a generating function of these numbers, 
we introduce a new set of variables 
$\bar{\bsx} = (\bar{x}_1,\bar{x}_2,\cdots)$, 
their power sums 
\beqnn
  \bar{p}_k = \sum_{i\ge 1}\bar{x}_i^k
\eeqnn
and their monomials 
\beqnn
  \bar{p}_\lambda 
  = \bar{p}_{\lambda_1}\bar{p}_{\lambda_2}\cdots 
\eeqnn
along with the variables in the case of simple Hurwitz numbers. 
Again, with the aid of the Frobenius formula 
(\ref{Frobenius-formula}), the generating function 
\beqnn
  Z(\bsx,\bar{\bsx}) 
  = \sum_{r=0}^\infty \sum_{d=0}^\infty \sum_{|\mu|=|\bar{\mu}|=d}
    H_d(\bar{\mu},\underbrace{1^{d-2}2,\cdots,1^{d-2}2}_r,\mu) 
    \frac{\beta^r}{r!}Q^d p_\mu\bar{p}_\mu 
\eeqnn
can be converted to a sum over all partitions as 
\beq
  Z(\bsx,\bar{\bsx}) 
  = \sum_{\lambda\in\calP}
    e^{\beta\kappa_\lambda/2}Q^{|\lambda|}
    s_\lambda(\bsx)s_\lambda(\bar{\bsx}).
\eeq

We now introduce the two sets 
$\bst = (t_1,t_2,\cdots)$ and 
$\bar{\bst} = (\bar{t}_1,\bar{t}_2,\cdots)$ 
of  time variables as 
\beqnn
  t_k = \frac{p_k}{k},\quad 
  \bar{t}_k = - \frac{\bar{p}_k}{k} 
\eeqnn
and consider the generating function 
\beq
  Z[\bst,\bar{\bst}] 
  = \sum_{\lambda\in\calP}
    e^{\beta\kappa_\lambda/2}Q^{|\lambda|}
    s_\lambda[\bst]s_\lambda[-\bar{\bst}] 
\label{Z[t,tbar]}
\eeq
of the new variables.  $Z[\bst,\bar{\bst}]$ is a tau function 
of the Toda hierarchy at a point of the lattice \cite{Okounkov00}
\footnote{A generalization of this tau function was first studied 
by Kharchev et al. \cite{KMMM93} in a different context.}.  
Reversing the sign of $\bar{\bst}$ is conventional 
in the formulation of the Toda hierarchy \cite{UT84} 
(cf. the fermionic representation of Toda tau functions 
reviewed in Section 3.4).  Any Toda tau function thereby 
becomes a tau function of the two-component KP hierarchy.  
Let us note that the the roles of $\bst$ and $\bar{\bst}$ 
in $Z[\bst,\bar{\bst}]$ can be interchanged 
by virtue of the identities 
\beq
  s_\lambda[\bst] = (-1)^{|\lambda|}s_{\tp{\lambda}}[-\bst],\quad 
  s_\lambda[-\bar{\bst}] = (-1)^{|\lambda|}s_{\tp{\lambda}}[\bar{\bst}]
\eeq
of the Schur functions and the property (\ref{kappa-antisymmetry}) 
of $\kappa_\lambda$.

\subsection{Cut-and-join operator}

The cut-and-join operator $M_0$ \cite{GJ97} may be thought of 
as an infinitesimal symmetry on the space of tau functions 
of the KP hierarchy \cite{KL06,MM08,Kazarian08}.  
In the KP time variables $\bst$, the cut-and-join operator reads 
\beq
  M_0 = \frac{1}{2}\sum_{j,k=1}^\infty 
        \left(klt_kt_l\frac{\rd}{\rd t_{k+l}} 
          + (k+l)t_{k+l}\frac{\rd^2}{\rd t_k\rd t_l}\right). 
\label{boson-M_0}
\eeq
The Schur functions $s_\lambda[\bst]$ are eigenfunctions 
of this operator with eigenvalues $\kappa_\lambda/2$:
\beq
  M_0 s_\lambda[\bst] = \frac{\kappa_\lambda}{2}s_\lambda[\bst]
\label{M_0Schur=...Schur}
\eeq
A combinatorial proof of this fact is presented 
in Zhou's paper \cite{Zhou03}.  As we shall remark in Section 3, 
the cut-and-join operator has a fermionic counterpart \cite{Okounkov00}, 
which leads to another proof of (\ref{M_0Schur=...Schur}).  

(\ref{M_0Schur=...Schur}) implies the identities 
\beqnn
  e^{\beta\kappa_\lambda/2}s_\lambda[\bst] 
  = e^{\beta M_0}s_\lambda[\bst]. 
\eeqnn
Therefore one can use $e^{\beta M_0}$ 
to recover $Z[\bst]$ and $Z[\bst,\bar{\bst}]$ 
from their ``initial values'' at $\beta = 0$ as 
\beqnn
  Z[\bst] = e^{\beta M_0}Z[\bst]|_{\beta=0}, \quad 
  Z[\bst,\bar{\bst}] = e^{\beta M_0}Z[\bst,\bar{\bst}]|_{\beta=0}.
\eeqnn
$Z[\bst]_{\beta=0}$ and $Z[\bst,\bar{\bst}]|_{\beta=0}$ 
can be calculated by the Cauchy identity
\beq
  \sum_{\lambda\in\calP}s_\lambda[\bst]s_\lambda[-\bar{\bst}]
  = \exp\left(- \sum_{k=1}^\infty kt_k\bar{t}_k\right) 
\label{Cauchy-identity}
\eeq
and the weighted homogeneity 
\beq
  s_\lambda[ct_1,c^2t_2,\cdots] 
  = c^{|\lambda|}s_\lambda[t_1,t_2,\cdots] 
\label{homogeneity}
\eeq
of Schur functions as 
\beqnn
  Z[\bst]|_{\beta=0}
  = \sum_{\lambda\in\calP}Q^{|\lambda|}
     s_\lambda[\bst]s_\lambda[1,0,0,\cdots]
  = e^{Qt_1}
\eeqnn
and 
\beqnn
  Z[\bst,\bar{\bst}]|_{\beta=0}
  = \sum_{\lambda\in\calP}Q^{|\lambda|}
     s_\lambda[\bst]s_\lambda[-\bar{\bst}]
  = \exp\left(- \sum_{k=1}^\infty Q^kkt_k\bar{t}_k\right). 
\eeqnn
One can thus derive the well known formula 
\cite{KL06,MM08,Kazarian08} 
\beq
  Z[\bst] = e^{\beta M_0}e^{Qt_1}
\label{Z[t]-M0}
\eeq
and its extension 
\beq
  Z[\bst,\bar{\bst}] 
  = e^{\beta M_0}\exp\left(- \sum_{k=1}^\infty Q^kkt_k\bar{t}_k\right) 
\label{Z[t,tbar]-M0}
\eeq
to double Hurwitz numbers.

\section{Fermionic representation of tau function}

\subsection{Two-dimensional free fermion system}

Let us introduce two-dimensional complex free fermion fields 
\beqnn
  \psi(z) = \sum_{n\in\ZZ} \psi_nz^{-n-1}, \quad 
  \psi^*(z) = \sum_{n\in\ZZ} \psi^*_nz^{-n}. 
\eeqnn
Note that we follow the notations of our previous work 
\cite{NT07,NT08} to use integers rather than half-integers 
for the labels of Fourier modes $\psi_n,\psi^*_n$. 
The Fourier modes satisfy the anti-commutation relations 
\beqnn
  \psi_m\psi^*_n + \psi^*_n\psi_m = \delta_{m+n,0}, \quad 
  \psi_m\psi_n + \psi_n\psi_m   = 0,\quad 
  \psi^*_m\psi^*_n + \psi^*_n\psi^*_m = 0. 
\eeqnn
The Fock space $\calH$ of ket vectors 
and its dual space $\calH^*$ of bra vectors 
are decomposed to charge-$s$ sectors $\calH_s,\calH^*_s$, 
$s \in \ZZ$.  Let $\langle s|$ and $|s\rangle$ denote 
the ground states in $\calH_s$ and $\calH^*_s$, namely, 
\beqnn
  \langle s| = \langle-\infty|\cdots\psi^*_{s-1}\psi^*_{s},\quad 
  |s\rangle = \psi_{-s}\psi_{-s+1}\cdots|-\infty\rangle, 
\eeqnn
which satisfy the annihilation conditions 
\beqnn
\begin{gathered}
  \psi_n|s\rangle = 0 \quad\mbox{for $n\ge -s$}, \quad 
  \psi^*_n|s\rangle = 0 \quad\mbox{for $n\ge s+1$},\\
  \langle s|\psi_n = 0 \quad\mbox{for $n\le -s-1$},\quad 
  \langle s|\psi^*_n = 0 \quad\mbox{for $n\le s$}. 
\end{gathered}
\eeqnn
Excited states can be labelled by partitions $\lambda 
= (\lambda_1,\lambda_2,\cdots,\lambda_n,0,0,\cdots)$ 
of arbitrary length as 
\beqnn
\begin{gathered}
  |\lambda,s\rangle 
 = \psi_{-\lambda_1-s}\cdots\psi_{-\lambda_n-s+n-1} 
   \psi^*_{s-n+1}\cdots\psi^*_{s}|s\rangle,\\
  \langle\lambda,s| 
 = \langle s|\psi_{-s}\cdots\psi_{-s+n-1}
   \psi^*_{\lambda_n+s-n+1}\cdots\psi^*_{\lambda_1+s}. 
\end{gathered}
\eeqnn
$|\lambda,s\rangle$ and $\langle\lambda,s|$ represent 
a state in which the semi-infinite subset 
$\{\lambda_i+s-i+1\}_{i=1}^\infty$ (sometimes referred to 
as the ``Maya diagram'') of the set $\ZZ$ 
of all ``energy levels'' are occupied by particles.   
These vectors give dual bases of of $\calH_s$ and 
$\calH^*_s$ in the sense that 
\beq
  \langle\lambda,r|\mu,s\rangle = \delta_{\lambda\mu}\delta_{rs}.
\eeq

The normal ordered fermion bilinears 
\beqnn
  {:}\psi_{-i}\psi^*_j{:} 
  = \psi_{-i}\psi^*_j - \langle 0|\psi_{-i}\psi^*_j|0\rangle, 
  \quad i,j \in \ZZ, 
\eeqnn
span the one-dimensional central extension $\widehat{\gl}(\infty)$ 
of the Lie algebra $\gl(\infty)$ of infinite matrices 
\cite{DJKM83,MJD-book}.  $\gl(\infty)$ consists of 
infinite matrices $A = (a_{ij})_{i,j\in\ZZ}$ of ``finite-band type'', 
namely, there is a positive integer $N$ (depending on $A$) 
such that $a_{ij} = 0 \quad \mbox{if $|i-j| > N$}$.  
For such a matrix $A \in \gl(\infty)$, the fermion bilinear 
\beqnn
  \widehat{A} = \sum_{i,j\in\ZZ}a_{ij}{:}\psi_{-i}\psi^*_j{:} 
\eeqnn
becomes a well-defined linear operator on the Fock space, 
and preserves the charge, namely, 
\beq
  \langle\lambda,r|\widehat{A}|\mu,s\rangle = 0 
  \quad\text{if $r \not= s$.}
\eeq
  Moreover, 
for two such matrices $A,B \in \gl(\infty)$, 
the associated operators $\widehat{A},\widehat{B}$ satisfy 
the commutation relation
\beq
  [\widehat{A},\widehat{B}] = \widehat{[A,B]} + \gamma(A,B) 
\eeq
with the $c$-number cocycle term 
\beq
  \gamma(A,B) = \Tr(A_{+-}B_{-+} - B_{+-}A_{-+}), 
\eeq
where $A_{\pm\mp}$ and $B_{\pm\mp}$ denote 
the following quarter blocks of $A,B$: 
\beqnn
\begin{gathered}
  A_{+-} = (a_{ij})_{i>0,\,j\le 0},\quad 
  A_{-+} = (a_{ij})_{i\le 0,\,j> 0}, \\
  B_{+-} = (b_{ij})_{i>0,\,j\le 0},\quad  
  B_{-+} = (b_{ij})_{i\le 0,\,j>0}. 
\end{gathered}
\eeqnn

\subsection{Special fermion bilinears}

The following fermion bilinears are building blocks 
of our Toda tau function: 
\beqnn
\begin{gathered}
  J_m = \sum_{n\in\ZZ} {:}\psi_{-n+m}\psi^*_n{:}, \quad m \in \ZZ,\\
  L_0 = \sum_{n\in\ZZ} n{:}\psi_{-n}\psi^*_n{:},\quad
  W_0 = \sum_{n\in\ZZ} n^2{:}\psi_{-n}\psi^*_n{:}.
\end{gathered}
\eeqnn
$J_m$'s span a $\mathrm{U}(1)$ current algebra. 
$L_0$ and $W_0$ are zero-modes of Virasoro and $W^{(3)}$ algebras.  
These fermion bilinears are associated with infinite matrices as 
\beq
  J_m = \widehat{\Lambda^m},\quad
  L_0 = \widehat{\Delta},\quad 
  W_0 = \widehat{\Delta^2}, 
\eeq
where $\Delta$ and $\Lambda$ are infinite matrices of the form 
\beqnn
  \Delta = (i\delta_{ij}), \quad 
  \Lambda = (\delta_{i+1,j}). 
\eeqnn

Let $J_{\pm}[\bst]$, $\bst = (t_1,t_2,\cdots)$, 
denote the special linear combinations 
\beqnn
  J_{+}[\bst] = \sum_{k=1}^\infty t_kJ_k, \quad 
  J_{-}[\bst] = \sum_{k=1}^\infty t_kJ_{-k}  
\eeqnn
of $J_m$'s.  Their exponentials act on 
the ground states $\langle s|$, $|s\rangle$ as 
\beq
  \langle s|e^{J_{+}[\bst]} 
  = \sum_{\lambda\in\calP}\langle\lambda,s|s_\lambda[\bst],\quad
  e^{J_{-}[\bst]}|s\rangle 
  = \sum_{\lambda\in\calP}s_\lambda[\bst]|\lambda,s\rangle, 
\label{e^J|s>}
\eeq
yielding Schur functions as matrix elements 
\cite{DJKM83,MJD-book}: 
\beq
  s_\lambda[\bst] 
  = \langle\lambda,s|e^{J_{+}[\bst]}|s\rangle 
  = \langle s|e^{J_{-}[\bst]}|\lambda,s\rangle. 
\label{Schur=<...>}
\eeq

Unlike other $J_m$'s, $J_0$ is diagonal with respect 
to $|\lambda,s\rangle$'s: 
\beq
  \langle\lambda,s|J_0|\mu,s\rangle = \delta_{\lambda\mu}s. 
\label{<ls|J_0|ms>}
\eeq
$L_0$ and $W_0$, too, are diagonal.  The diagonal matrix elements 
can be calculated as follows. 

\begin{lemma}
\begin{align}
  \langle\lambda,s|L_0|\mu,s\rangle   
    &= \delta_{\lambda\mu}
       \left(|\lambda| + \frac{s(s+1)}{2}\right), 
  \label{<ls|L_0|ms>}\\
  \langle\lambda,s|W_0|\mu,s\rangle
    &= \delta_{\lambda\mu}
       \left(\kappa_\lambda + (2s+1)|\lambda| 
        + \frac{s(s+1)(2s+1)}{6}\right). 
  \label{<ls|W_0|ms>}
\end{align}
\end{lemma}

\begin{proof}
Assuming that $s \ge 0$, one can calculate 
the diagonal matrix elements as 
\beqnn
\begin{aligned}
  \langle\lambda,s|L_0|\lambda,s\rangle 
  &= \sum_{i=1}^\infty(\lambda_i+s-i+1) - \sum_{i=1}^\infty (-i+1)
     \quad \mbox{(heuristic expression)}\\
  &= \sum_{i=1}^\infty((\lambda_i+s-i+1) - (s-i+1)) + \sum_{k=0}^s k 
     \quad \mbox{(re-summed)}\\
  &= |\lambda| + \frac{s(s+1)}{2} 
\end{aligned}
\eeqnn
and 
\beqnn
\begin{aligned}
  \langle\lambda,s|W_0|\lambda,s\rangle 
  &= \sum_{i=1}^\infty(\lambda_i+s-i+1)^2 - \sum_{i=1}^\infty(-i+1)^2 
     \quad \mbox{(heuristic expression)}\\
  &= \sum_{i=1}^\infty(\lambda_i+s-i+1)^2 - (s-i+1)^2) 
     + \sum_{k=0}^s k^2 \quad \mbox{(re-summed)}\\
  &= \kappa_\lambda + (2s+1)|\lambda| 
     + \frac{s(s+1)(2s+1)}{6}. 
\end{aligned}
\eeqnn
In the case where $s < 0$, we have only to replace 
the intermediate sums over $k$ as 
\beqnn
  \sum_{k=0}^s k \to - \sum_{k=s+1}^0 k, \quad 
  \sum_{k=0}^s k^2 \to - \sum_{k=s+1}^0 k^2, 
\eeqnn
ending up with the same final expression 
of the matrix elements.  
\end{proof}

\subsection{Toda Tau function for double Hurwitz numbers}

A general tau function of the Toda hierarchy 
has the fermionic expression 
\beqnn
  \tau(s,\bst,\bar{\bst}) 
  = \langle s|e^{J_{+}[\bst]}ge^{-J_{-}[\bar{\bst}]}|s\rangle, 
\eeqnn
where $g$ is an element of $\widehat{\GL}(\infty)$ \cite{Takebe91}. 
Inserting the aforementioned expansion (\ref{e^J|s>}), 
one can expand this function as 
\beq
  \tau(s,\bst,\bar{\bst}) 
  = \sum_{\lambda,\mu\in\calP}
    \langle\lambda,s|g|\mu,s\rangle 
    s_\lambda[\bst]s_\mu[-\bar{\bst}]. 
\label{general-tau}
\eeq

Following Okounkov \cite{Okounkov00}, 
we now consider the special case where 
\beq
  g = e^{\beta W_0/2}Q^{L_0} 
    = Q^{L_0}e^{\beta W_0/2}. 
\label{Hurwitz-g}
\eeq

\begin{theorem}
The tau function determined by (\ref{Hurwitz-g}) 
can be expanded as 
\beq
  \tau(s,\bst,\bar{\bst}) 
  = e^{\beta s(s+1)(2s+1)/12}Q^{s(s+1)/2} 
    \sum_{\lambda\in\calP}
    e^{\beta\kappa_\lambda/2}(e^{\beta(s+1/2)}Q)^{|\lambda|}
    s_\lambda[\bst]s_\lambda[-\bar{\bst}].
\label{Hurwitz-tau}
\eeq
\end{theorem}

\begin{proof}
The properties (\ref{<ls|L_0|ms>}) and (\ref{<ls|W_0|ms>}) 
of $L_0$ and $W_0$ imply that $g$ is also diagonalized 
on the basis $|\lambda,s\rangle$ of the Fock space. 
The diagonal matrix elements take such a form as 
\beqnn
  \langle\lambda,s|g|\lambda,s\rangle 
  = \exp\left(\frac{\beta}{2}\left(\kappa_\lambda+(2s+1)|\lambda|
       + \frac{s(s+1)(2s+1)}{6}\right)\right) 
    Q^{|\lambda|+s(s+1)/2}. 
\eeqnn
The tau function in question can be thereby expanded 
as (\ref{Hurwitz-tau}) shows. 
\end{proof}

This is a restatement of Okounkov's result \cite{Okounkov00}. 
One can rewrite (\ref{Hurwitz-tau}) as 
\beq
  \tau(s,\bst,\bar{\bst}) 
  = e^{\beta s(s+1)(2s+1)/12}Q^{s(s+1)/2} 
    Z_{\beta,e^{\beta(s+1/2)}Q}[\bst,\bar{\bst}], 
\label{Hurwitz-tau-Z}
\eeq
where $Z_{\beta,Q}[\bst,\bar{\bst}]$ denotes 
the generating function (\ref{Z[t,tbar]}) 
with the parameters $\beta$ and $Q$ being 
explicitly indicated.  Thus, apart from 
some numerical factors depending on $s$, 
the tau function coincides with the generating function 
of double Hurwitz numbers.  Note that the $s$-dependence 
shows up in the generating function as 
the multiplier $e^{\beta(s+1/2)}$ of the parameter $Q$. 

Let us conclude this section with a few comments 
on the cut-and-join operator (\ref{boson-M_0}). 
The cut-and-join operator corresponds to the fermion bilinear 
\beq
  M_0 = \frac{1}{2}\sum_{n\in\ZZ} 
        \left(n - \frac{1}{2}\right)^2{:}\psi_{-n}\psi^*_n{:} 
      = \frac{W_0}{2} - \frac{L_0}{2} + \frac{J_0}{8}. 
\label{fermion-M_0}
\eeq
One can readily see from (\ref{<ls|J_0|ms>}), 
(\ref{<ls|L_0|ms>}) and (\ref{<ls|W_0|ms>}) 
that this operator acts on $|\lambda,s\rangle$'s as 
\beq
  M_0|\lambda,s\rangle 
  = \left(\frac{\kappa_\lambda}{2} + s|\lambda| 
       + \frac{4s^3-s}{24}\right)|\lambda,s\rangle.
\eeq
In particular, $|\lambda\rangle = |\lambda,0\rangle$ 
is an eigenstate with eigenvalue $\kappa_\lambda/2$.  
Note that $s$-dependent extra terms show up 
in the charge-$s$ sector.  

Fermion bilinears of this type can be converted 
to differential (or ``bosonic'') operators 
by the boson-fermion correspondence \cite{DJKM83,MJD-book}.  
In a generating functional form, the normal-ordered product 
\beqnn
  {:}\psi(z)\psi^*(w){:} = \psi(z)\psi^*(w) - \frac{1}{z-w} 
  \quad (|z|<|w|) 
\eeqnn
of the fermion fields corresponds to 
the two-variable vertex operator 
\begin{multline*}
  X(z,w) \\
  = \frac{1}{z-w}\left(
      \left(\frac{z}{w}\right)^s
      \exp\left(\sum_{k=1}^\infty t_k(z^k-w^k)\right)
      \exp\left(-\sum_{k=1}^\infty\frac{z^{-k}-w^{-k}}{k}
           \frac{\rd}{\rd t_k}\right) 
      - 1\right)
\end{multline*}
as 
\beq
  \langle s|e^{J_{+}[\bst]}{:}\psi(z)\psi^*(w){:}
  = X(z,w)\langle s|e^{J_{+}[\bst]}. 
\eeq
A similar relation holds for $e^{-J_{-}[\bst]}|s\rangle$ 
and leads to bosonization with respect to $\bar{\bst}$ \cite{TT95}, 
though we omit details here. 

$\Gamma(z,w)$ can be expanded in powers of $z-w$, 
and the coefficients of this expansion give 
a bosonic representation of fermion bilinears.  
For $L_0,J_0$ and $W_0$, this bosonic representation reads 
\beq
  L_0 = \sum_{k=1}^\infty kt_k\frac{\rd}{\rd t_k} 
        + \frac{s(s+1)}{2}, \quad 
  J_0 = s 
\eeq
and 
\begin{multline}
  W_0 = \sum_{k,l=1}^\infty\left( klt_kt_l\frac{\rd}{\rd t_{k+l}}
         + (k+l)t_{k+l}\frac{\rd^2}{\rd t_k\rd t_l}\right)\\
    + (2s+1)\sum_{k=1}^\infty kt_k\frac{\rd}{\rd t_k}
    + \frac{s(s+1)(2s+1)}{6}. 
\end{multline}
(\ref{fermion-M_0}) is thus bosonized as 
\begin{multline}
  M_0 
  = \frac{1}{2}\sum_{k,l=1}^\infty
    \left( klt_kt_l\frac{\rd}{\rd t_{k+l}}
      + (k+l)t_{k+l}\frac{\rd^2}{\rd t_k\rd t_l}\right)\\
    + s\sum_{k=1}^\infty kt_k\frac{\rd}{\rd t_k}
    + \frac{4s^3-s}{24}. 
\end{multline}
In the charge-0 sector, this reduces to 
the cut-and-join operator (\ref{boson-M_0}).

\section{Generalized string equations for double Hurwitz numbers}

\subsection{Notations for difference operators}

Building blocks of the Lax formalism of 
the Toda hierarchy are one-dimensional difference operators 
in the lattice coordinate $s$ \cite{UT84}.  
Those operators are linear combinations 
of the shift operators $e^{n\rd_s}$, 
$e^{n\rd_s}f(s) = f(s+n)$.  
Although a genuine difference operator is a finite 
linear combination 
\beqnn
  A = \sum_{n=M}^N a_n(s)e^{n\rd_s} 
  \quad \mbox{(operator of $[M,N]$-type)}
\eeqnn
of the shift operators, one can consider 
a semi-infinite linear combination of the form 
\beqnn
  A = \sum_{n=-\infty}^N a_n(s)e^{n\rd_s}
  \quad \mbox{(operator of $(-\infty,N]$ type}) 
\eeqnn
or 
\beqnn
  A = \sum_{n=M}^\infty a_n(s)e^{n\rd_s}
  \quad \mbox{(operator of $[M,\infty)$ type)}
\eeqnn
as well, which amount to pseudo-differential operators 
in the Lax formalism of the KP hierarchy \cite{DJKM83}.  
Let $(\quad)_{\ge 0}$ and $(\quad)_{<0}$ denote 
the truncation operation 
\beqnn
  (A)_{\ge 0} = \sum_{n\ge 0}a_n(s)e^{n\rd_s}, \quad 
  (A)_{<0} = \sum_{n<0}a_n(s)e^{n\rd_s}. 
\eeqnn

Difference operators are in one-to-one correspondence 
with $\ZZ\times\ZZ$ matrices.  Firstly, the $n$-th shift operator
$e^{n\rd_s}$ corresponds to the shift matrix 
\beqnn
  \Lambda^n = (\delta_{i,j-n})_{i,j\in\ZZ}. 
\eeqnn
Secondly, the multiplication operator $a(s)$ 
amounts to the diagonal matrix 
\beqnn
  \diag(a(s)) = (a(i)\delta_{ij})_{i,j\in\ZZ}. 
\eeqnn
In particular, the multiplication operator $s$ corresponds to 
\beqnn
  \Delta = \diag(s) = (i\delta_{ij})_{i,j\in\ZZ}. 
\eeqnn
Consequently, a general difference operator 
\beqnn
  A = A(s,e^{\rd_s}) = \sum_{n} a_n(s)e^{n\rd_s}
\eeqnn
is converted to the infinite matrix 
\beqnn
  A(\Delta,\Lambda) 
  = \sum_{n}\diag(a_n(s))\Lambda^n 
  = \sum_{n} (a_n(i)\delta_{i,j-n})_{i,j\in\ZZ}. 
\eeqnn
Occasionally, it might be more convenient to write 
a difference operator in an anti-normal-ordered form as 
\beqnn
  B(e^{\rd_s},s) = \sum_{n} e^{n\rd_s}b_n(s). 
\eeqnn
In that case, the corresponding infinite matrix reads 
\beqnn
  B(\Lambda,\Delta ) = \sum_{n} \Lambda^n\diag(b_n(s)). 
\eeqnn

\subsection{Lax and Orlov-Schulman operators}

The Lax formalism of the Toda hierarchy uses 
two Lax operators $L,\bar{L}$ of type $(-\infty,1]$ 
and $[1,\infty)$.   Actually, from the point of view of symmetry, 
it is better to consider $L$ and $\bar{L}^{-1}$ 
rather than $L$ and $\bar{L}$.  These operators 
admit freedom of gauge transformations.  
In the gauge where $L$ is {\it monic} (namely, 
the leading coefficients is equal to $1$), 
$L$ and $\bar{L}^{-1}$ can be expressed as 
\beqnn
\begin{aligned}
  L &= e^{\rd_s} + \sum_{n=1}^\infty u_ne^{(1-n)\rd_s},\\
  \bar{L}^{-1} &= \bar{u}_0e^{-\rd_s} 
     + \sum_{n=1}^\infty \bar{u}_n e^{(n-1)\rd_s}. 
\end{aligned}
\eeqnn
The coefficients $u_n$ and $\bar{u}_n$ are functions 
of $s$ and the time variables $\bst,\bar{\bst}$, and 
written as $u_n(s,\bst,\bar{\bst})$ and $\bar{u}_n(s,\bst,\bar{\bst})$ 
if we do not suppress the independent variables.  
$L$ and $\bar{L}$ satisfy the Lax equations 
\beq
\begin{gathered}
  \frac{\rd L}{\rd t_n} = [B_n,L], \quad 
  \frac{\rd L}{\rd\bar{t}_n} = [\bar{B}_n,L], \\
  \frac{\rd\bar{L}}{\rd t_n} = [B_n,\bar{L}],\quad
  \frac{\rd\bar{L}}{\rd\bar{t}_n} = [\bar{B}_n,\bar{L}], 
\end{gathered}
\label{LLbar-Laxeq}
\eeq
where $B_n$ and $\bar{B}_n$ are defined as 
\beqnn
  B_n = (L^n)_{\ge 0}, \quad 
  \bar{B}_n = (\bar{L}^{-n})_{<0}. 
\eeqnn

To formulate the generalized string equations, 
we need another pair of difference operators $M,\bar{M}$, 
namely, the Orlov-Schulman operators \cite{TT93}. 
These operators, too, satisfy the Lax equations 
\beq
\begin{gathered}
  \frac{\rd M}{\rd t_n} = [B_n,M], \quad 
  \frac{\rd M}{\rd\bar{t}_n} = [\bar{B}_n,M], \\
  \frac{\rd\bar{M}}{\rd t_n} = [B_n,\bar{M}],\quad
  \frac{\rd\bar{M}}{\rd\bar{t}_n} = [\bar{B}_n,\bar{M}] 
\end{gathered}
\label{MMbar-Laxeq}
\eeq
of the same form as the Lax operators do, 
and are related to the Lax operators 
by the twisted canonical commutation relations 
\beq
  [L,M] = L, \quad [\bar{L},\bar{M}] = \bar{L}. 
\label{LM-CCR}
\eeq

By ``generalized string equations'' we mean 
equations of the form 
\beq
  C(L,M) = \bar{C}(\bar{L},\bar{M}), 
\label{C(L,M)=Cbar(Lbar,Mbar)}
\eeq
where $C(L,M)$ and $\bar{C}(\bar{L},\bar{M})$ are 
(possibly infinite) linear combinations of monomials 
of $L,M$ and $\bar{L},\bar{M}$ with constant coefficients.  
The following lemma \cite{NTT95,Takasaki96} explains 
an origin of generalized string equations. 

\begin{lemma}
If the fermion bilinears $\widehat{C(\Lambda,\Delta)}$ 
and $\widehat{\bar{C}(\Lambda,\Delta)}$ are intertwined 
by a $\widehat{GL}(\infty)$ element $g$ as 
\beq
  \widehat{C(\Lambda,\Delta)}g 
  = g\widehat{\bar{C}(\Lambda,\Delta)} 
\label{C(Lam,Del)=Cbar(Lam,Del)}, 
\eeq
then the Lax and Orlov-Schulman operators satisfy 
(\ref{C(L,M)=Cbar(Lbar,Mbar)}). 
\end{lemma}

\subsection{Intertwining relations}

Let us now consider the case of the solution determined by 
the $\widehat{GL}(\infty)$ element (\ref{Hurwitz-g}).  
We seek intertwining relations in the following special form: 
\beqnn
  J_kg = g\widehat{\bar{C}(\Lambda,\Delta)}, \quad 
  \widehat{C(\Lambda,\Delta)}g = gJ_{-k}, \quad 
  k = 1,2,\cdots. 
\eeqnn

\begin{lemma}
$J_m$'s are transformed by the adjoint action 
of $Q^{L_0}$ and $e^{\beta W_0/2}$ as 
\beq
\begin{gathered}
  Q^{L_0}J_mQ^{-L_0} = Q^{-m}J_m,\\
  e^{-\beta W_0/2}J_me^{\beta W_0/2} 
    = e^{-\beta m^2/2}\sum_{n\in\ZZ}e^{\beta mn}{:}\psi_{-n+m}\psi^*_n{:}.
\end{gathered}
\label{J_m-transform}
\eeq
\end{lemma}

\begin{proof}
Let us note the fundamental commutation relations 
\beqnn
  [L_0,\psi_n] = - n\psi_n,\quad 
  [L_0,\psi^*_n] = - n\psi^*_n 
\eeqnn
and 
\beqnn
  [W_0,\psi_n] = n^2\psi_n,\quad 
  [W_0,\psi^*_n] = - n^2\psi^*_n
\eeqnn
that follow from the definition of $L_0$ and $W_0$.  
These commutation relations can be exponentiated as 
\beqnn
  Q^{L_0}\psi_nQ^{-L_0} = Q^{-n}\psi_n,\quad 
  Q^{L_0}\psi^*_nQ^{-L_0} = Q^{-n}\psi^*_n 
\eeqnn
and 
\beqnn
  e^{-\beta W_0/2}\psi_ne^{\beta W_0/2} 
  = e^{-\beta n^2/2}\psi_n, \quad 
  e^{-\beta W_0/2}\psi^*_ne^{\beta W_0/2}
  = e^{\beta n^2/2}\psi^*_n. 
\eeqnn
Consequently, the exponentiated operators $Q^{L_0}$ 
and $e^{\beta W_0/2}$ act on the basis ${:}\psi_{-i}\psi^*_j{:}$ 
of $\gl(\infty)$ as 
\beqnn
  Q^{L_0}{:}\psi_{-i}\psi^*_j{:}Q^{-L_0} 
  = Q^{i-j}{:}\psi_{-i}\psi^*_j{:}
\eeqnn
and 
\beqnn
  e^{-\beta W_0/2}{:}\psi_{-i}\psi^*_j{:}e^{\beta W_0/2}
  = e^{-\beta(i^2-j^2)/2}{:}\psi_{-i}\psi^*_j{:}. 
\eeqnn
One can thereby derive (\ref{J_m-transform}) as 
\beqnn
  Q^{L_0}J_mQ^{-L_0} 
  = \sum_{n\in\ZZ}Q^{(n-m)-n}{:}\psi_{-n+m}\psi^*_n{:} 
  =  Q^{-m}J_m
\eeqnn
and 
\beqnn
\begin{aligned}
  e^{-\beta W_0/2}J_me^{\beta W_0/2} 
  &= \sum_{n\in\ZZ}
      e^{-\beta((n-m)^2-n^2)/2}{:}\psi_{-n+m}\psi^*_n{:}\\
  &= e^{-\beta m^2/2}
    \sum_{n\in\ZZ}e^{\beta mn}{:}\psi_{-n+m}\psi^*_n{:}.
\end{aligned}
\eeqnn
\end{proof}

\begin{lemma}
$J_{\pm k}$'s are connected with the fermion bilinears 
$\widehat{\Lambda^ke^{\beta k\Delta}}$ 
and $\widehat{\Lambda^{-k}e^{\beta k\Delta}}$ 
by the $\widehat{\GL}(\infty)$ element (\ref{Hurwitz-g}) as 
\begin{gather}
  J_kg = gQ^ke^{-\beta k^2/2}\widehat{\Lambda^ke^{\beta k\Delta}}, 
  \label{Jg=gCbar}\\
  gJ_{-k} = Q^ke^{\beta k^2/2}\widehat{\Lambda^{-k}e^{\beta k\Delta}}g. 
  \label{gJ=Cg}
\end{gather}
\end{lemma}

\begin{proof}
Using the relations (\ref{J_m-transform}) in the previous lemma, 
one can calculate $g^{-1}J_kg$ as 
\beqnn
\begin{aligned}
  g^{-1}J_kg 
  &= e^{-\beta W_0/2}Q^{-L_0}J_kQ^{L_0}e^{\beta W_0/2}\\
  &= Q^ke^{-\beta W_0/2}J_ke^{\beta W_0/2}\\
  &= Q^ke^{-\beta k^2/2} 
     \sum_{n\in\ZZ}e^{\beta kn}{:}\psi_{-n+k}\psi^*_n{:}. 
\end{aligned}
\eeqnn
Since the last sum can be rewritten as 
\beqnn
  \sum_{n\in\ZZ}e^{\beta kn}{:}\psi_{-n+k}\psi^*_n{:} 
  = \widehat{\Lambda^ke^{\beta k\Delta}}, 
\eeqnn
the first intertwining relation (\ref{Jg=gCbar}) follows.  
In the same way, one can calculate $gJ_{-k}g^{-1}$ as 
\beqnn
\begin{aligned}
  gJ_{-k}g^{-1} 
  &= e^{\beta W_0/2}Q^{L_0}J_{-k}Q^{-L_0}e^{-\beta W_0/2}\\
  &= Q^ke^{\beta W_0/2}J_{-k}e^{-\beta W_0/2}\\
  &= Q^ke^{\beta k^2/2}
     \sum_{n\in\ZZ}e^{\beta kn}{:}\psi_{-n-k}\psi^*_n{:}, 
\end{aligned}
\eeqnn
which implies the second intertwining relation (\ref{gJ=Cg}). 
\end{proof}

\subsection{Generalized string equations}

\begin{theorem}
The Lax and Orlov-Schulman operators of 
the tau function (\ref{Hurwitz-tau}) satisfy 
the generalized string equations 
\beq
  L^k = Q^ke^{-\beta k^2/2}\bar{L}^ke^{\beta k\bar{M}}, \quad 
  \bar{L}^{-k} = Q^ke^{\beta k^2/2}L^{-k}e^{\beta kM} 
\label{str-eq(k)}
\eeq
for $k = 1,2,\cdots$.  Moreover, these equations 
can be derived from the first two ($k = 1$) equations
\beq
  L = Qe^{-\beta/2}\bar{L} e^{\beta\bar{M}},\quad 
  \bar{L}^{-1} = Qe^{\beta/2}L^{-1}e^{\beta M}. 
\label{str-eq}
\eeq
\end{theorem}

\begin{proof}
The first  part is a consequence of (\ref{Jg=gCbar}) 
and (\ref{gJ=Cg}).  Let us show the second part.  
The $k$-th power of the first equation of (\ref{str-eq}) reads 
\beqnn
  L^k = Q^ke^{-\beta k/2} (\bar{L} e^{\beta\bar{M}})^k. 
\eeqnn
The commutation equation of $\bar{L}$ and $\bar{M}$ 
in (\ref{LM-CCR}) implies that 
\beqnn
  [\bar{M},\cdots,[\bar{M},[\bar{M},\bar{L}]]\cdots] 
  \; \mbox{($k$-fold commutator)}\;  
  = (-1)^k\bar{L} 
\eeqnn
for $k = 1,2,\cdots$, so that 
\beqnn
  e^{\beta\bar{M}}\bar{L} e^{-\beta\bar{M}}
  = \bar{L} 
  + \sum_{k=1}^\infty \frac{\beta^k}{k!}
      [\bar{M},\cdots,[\bar{M},[\bar{M},\bar{L}]]\cdots] 
  = e^{-\beta}\bar{L}. 
\eeqnn
Using this relation repeatedly, one can move 
$e^{\beta\bar{M}}$'s in $(\bar{L}e^{\beta\bar{M}})^k$ 
to the rightmost position as 
\beqnn
\begin{aligned}
  (\bar{L} e^{\beta\bar{M}})^k 
  &= \bar{L} e^{\beta\bar{M}}(\bar{L} 
     e^{\beta\bar{M}})^{k-1}e^{\beta\bar{M}}\\
  &= e^{-\beta k}\bar{L}^2 
     e^{\beta\bar{M}}(\bar{L} e^{\beta\bar{M}})^{k-2}
     e^{2\beta\bar{M}}\\ 
  &= e^{-\beta k-\beta(k-1)}\bar{L}^3 
     e^{\beta\bar{M}}(e^{\beta\bar{M}}\bar{L})^{k-3}
     e^{3\beta\bar{M}}\\
  &= \cdots\\
  &= e^{-\beta k(k-1)/2}\bar{L}^k e^{\beta k\bar{M}}.
\end{aligned}
\eeqnn
Thus the first equation of (\ref{str-eq(k)}) follows. 
The second equation of (\ref{str-eq(k)}), too, can be 
derived from (\ref{str-eq}) in the same way. 
\end{proof}

Thus, in contrast with two-dimensional quantum gravity 
\cite{AvM92} and $c = 1$ string theory 
\cite{DMP93,HOP94,EK94,Takasaki95,NTT95,Takasaki96}, 
the generalized string equations contain 
the exponential terms $e^{\beta M},e^{\beta\bar{M}}$.  
These terms stem from the fermion bilinears 
$\widehat{\Lambda^{\pm k}e^{\beta k\Delta}}$ 
in (\ref{Jg=gCbar}) and (\ref{gJ=Cg}).  
Ferimion blinears of a similar form are also used 
in the study of integrable structures 
of the melting crystal model \cite{NT07,NT08}.  
A common algebraic background of these fermion bilinears 
is the quantum torus algebra (with parameter $q$) 
spanned by the infinite matrices 
\beq
  v^{(k)}_m = q^{-km/2}\Lambda^m q^{k\Delta}, \quad 
  k,m \in \ZZ. 
\eeq
In the melting crystal model \cite{NT07,NT08}, 
a central extension of this Lie algebra is realized 
by the fermion bilinears 
\beq
  V^{(k)}_m = q^{-km}\sum_{n\in\ZZ}q^{kn}{:}\psi_{-n+m}\psi^*_n{:}. 
\eeq
The Lax and Orlov-Schulman operators give (two copies of) 
yet another kind of realization by the difference operators 
\beq
  V^{(k)}_m = q^{-km/2}L^m q^{kM} 
\eeq
and 
\beq
  \bar{V}^{(k)}_m = q^{-km/2}\bar{L}^mq^{k\bar{M}}. 
\eeq
If $q = e^\beta$, the exponentials $e^{\beta M},e^{\beta\bar{M}}$ 
belong to this Lie algebra.

\section{Classical limit of generalized string equations}

\subsection{Dispersionless Toda hierarchy}

In a naive sense \cite{TT91}, the classical limit 
of the Toda hierarchy can be obtained 
by replacing the shift operator $e^{\rd_s}$ 
by a new variable $p$.  The difference operators 
$L,M,\bar{L},\bar{M}$ thus turn into Laurent series 
of $p$ of the form 
\beqnn
\begin{gathered}
  \calL = p + \sum_{n=1}^\infty u_n p^{1-n},\\
  \bar{\calL}^{-1} = \bar{u}_0 p^{-1} 
     + \sum_{n=1}^\infty \bar{u}_n p^{n-1},\\
  \calM = \sum_{n=1}^\infty nt_n\calL^n 
      + s + \sum_{n=1}^\infty v_n\calL^{-n},\\
 \bar{\calM} = - \sum_{n=1}^\infty n\bar{t}_n\bar{\calL}^{-n} 
      + s + \sum_{n=1}^\infty \bar{v}_n\bar{\calL}^n  
\end{gathered}
\eeqnn
that are referred to as the ``Lax and Orlov-Schulman functions''.  

As difference operators are replaced by Laurent series, 
commutators of difference operators turn into 
Poisson brackets by the rule 
\beqnn
  [e^{\rd_s},s] = e^{\rd_s} \;\to\;
  \{p,s\} = s. 
\eeqnn
Accordingly, Poisson brackets of functions 
of $p$ and $s$ are defined as 
\beqnn
  \{F,G\} = p\left(\frac{\rd F}{\rd p}\frac{\rd G}{\rd s} 
          - \frac{\rd F}{\rd s}\frac{\rd G}{\rd p}\right). 
\eeqnn
The Lax equations and the twisted canonical commutation relations 
are redefined with respect to the Poisson bracket as 
\beq
\begin{gathered}
  \frac{\rd \calL}{\rd t_n} = \{\calB_n,\calL\},\quad 
  \frac{\rd \calL}{\rd\bar{t}_n} = \{\bar{\calB}_n,\calL\}, \\
  \frac{\rd\bar{\calL}}{\rd t_n} = \{\calB_n,\bar{\calL}\},\quad
  \frac{\rd\bar{\calL}}{\rd\bar{t}_n} = \{\bar{\calB}_n,\bar{\calL}\}, \\
  \frac{\rd \calM}{\rd t_n} = \{\calB_n,\calM\},\quad 
  \frac{\rd \calM}{\rd\bar{t}_n} = \{\bar{\calB}_n,\calM\}, \\
  \frac{\rd\bar{\calM}}{\rd t_n} = \{\calB_n,\bar{\calM}\},\quad
  \frac{\rd\bar{\calM}}{\rd\bar{t}_n} = \{\bar{\calB}_n,\bar{\calM}\} 
\end{gathered}
\label{dToda-Laxeq}
\eeq
and 
\beq
  \{\calL,\calM\} = \calL, \quad \{\bar{\calL},\bar{\calM}\} = \bar{\calL}. 
\label{dToda-CCR}
\eeq
$\calB_n$ and $\bar{\calB}_n$ are defined 
by seemingly the same formulae 
\beqnn
  \calB_n = (\calL^n)_{\ge 0}, \quad 
  \bar{\calB}_n = (\bar{\calL}^{-n})_{<0} 
\eeqnn
as in the previous case, but the notations 
$(\quad)_{\ge 0}$ and $(\quad)_{<0}$ now stand for 
projection operators on the space of Laurent series, 
namely, 
\beqnn
  \left(\sum_{n}a_np^n\right)_{\ge 0} = \sum_{n\ge 0}a_np^n,\quad 
  \left(\sum_{n}a_np^n\right)_{<0} = \sum_{n<0}a_np^n. 
\eeqnn
These equations are fundamental constituents of 
the ``dispersionless Toda hierarchy''.

\subsection{$\hbar$-dependent Toda hierarchy and generalized string equations}

The foregoing procedure replacing $e^{\rd_s} \to p$ 
can be justified as a kind of classical limit 
in an $\hbar$-dependent formulation 
of the Toda hierarchy \cite{TT93}.  

In the $\hbar$-dependent formulation, 
$e^{\rd_s}$ is replaced by $e^{\hbar\rd_s}$.  
The ``Planck constant'' $\hbar$ thus plays 
the role of lattice spacing.  
The Lax and Orlov-Schulman operators are expanded 
in powers of $e^{\hbar\rd_s}$ as 
\beqnn
\begin{gathered}
  L = e^{\hbar\rd_s} + \sum_{n=1}^\infty u_ne^{(1-n)\hbar\rd_s},\\
  \bar{L}^{-1} = \bar{u}_0e^{-\hbar\rd_s} 
     + \sum_{n=1}^\infty \bar{u}_n e^{(n-1)\hbar\rd_s},\\
  M = \sum_{n=1}^\infty nt_nL^n 
      + s + \sum_{n=1}^\infty v_nL^{-n},\\
 \bar{M} = - \sum_{n=1}^\infty n\bar{t}_n\bar{L}^{-n} 
      + s + \sum_{n=1}^\infty \bar{v}_n\bar{L}^n. 
\end{gathered}
\eeqnn
The Lax equations and the twisted canonical commutation relations 
take an $\hbar$-dependent form as 
\beq
\begin{gathered}
  \hbar\frac{\rd L}{\rd t_n} = [B_n,L],\quad 
  \hbar\frac{\rd L}{\rd\bar{t}_n} = [\bar{B}_n,L], \\
  \hbar\frac{\rd\bar{L}}{\rd t_n} = [B_n,\bar{L}],\quad
  \hbar\frac{\rd\bar{L}}{\rd\bar{t}} = [\bar{B}_n,\bar{L}],\\
  \hbar\frac{\rd M}{\rd t_n} = [B_n,M],\quad
  \hbar\frac{\rd M}{\rd\bar{t}_n} = [\bar{B}_n,M], \\
  \hbar\frac{\rd\bar{M}}{\rd t_n} = [B_n,\bar{M}],\quad
  \hbar\frac{\rd\bar{M}}{\rd\bar{t}_n} = [\bar{B}_n,\bar{M}]
\end{gathered}
\eeq
and 
\beq
  [L,M] = \hbar L, \quad [\bar{L},\bar{M}] = \hbar\bar{L}. 
\label{hbar-LM-CCR}
\eeq

If the coefficient $u_n,\bar{u}_n,v_n,\bar{v}_n$ 
(which are functions of $\hbar,s,\bst,\bar{\bst}$) 
have a smooth classical limit as 
\beq
  u_n^{(0)} = \lim_{\hbar\to 0}u_n,\quad 
  \bar{u}_n^{(0)} = \lim_{\hbar\to 0}\bar{u}_n,\quad 
  v_n^{(0)} = \lim_{\hbar\to 0}v_n,\quad 
  \bar{v}_n^{(0)} = \lim_{\hbar\to 0}\bar{v}_n, 
\label{uv-cl-condition}
\eeq
one can define the Lax and Orlov-Schulman functions as 
\beqnn
\begin{gathered}
  \calL  = p + \sum_{n=1}^\infty u_n^{(0)} p^{1-n},\\
  \bar{\calL}^{-1} = \bar{u}_0^{(0)} p^{-1} 
     + \sum_{n=1}^\infty \bar{u}_n^{(0)} p^{n-1},\\
  \calM = \sum_{n=1}^\infty nt_n\calL^{n} 
      + s + \sum_{n=1}^\infty v^{(0)}_n\calL^{-n},\\
  \bar{\calM} = - \sum_{n=1}^\infty n\bar{t}_n\bar{\calL}^{-n} 
      + s + \sum_{n=1}^\infty \bar{v}_n^{(0)}\bar{\calL}^{n}. 
\end{gathered}
\eeqnn
These Lax and Orlov-Schulman functions satisfy 
the Lax equations (\ref{dToda-Laxeq}) and 
the twisted canonical Poisson relations (\ref{dToda-CCR}). 

In this $\hbar$-dependent formulation, 
generalized string equations (\ref{C(L,M)=Cbar(Lbar,Mbar)}) 
are modified as 
\beq
  C(L,\hbar^{-1}M) = \bar{C}(\bar{L},\hbar^{-1}\bar{M}), 
\label{hbar-C(L,M)=Cbar(Lbar,Mbar)}
\eeq
namely, $M$ and $\bar{M}$ are multiplied by $\hbar^{-1}$ 
\cite{NTT95,Takasaki96}.  
Let us explain the underlying mechanism briefly.  
The Lax and Orlov-Schulman operators are connected 
with the matrices $\Lambda$ and $\Delta$ 
by the so called ``dressing operators'' \cite{TT93}.  
The twisted canonical commutation relations 
are thereby derived from the commutation relation 
\beqnn
  [\Lambda,\Delta] = \Lambda
\eeqnn
of these matrices.  In the $\hbar$-dependent formulation, 
they take the form
\beqnn
  [L,\hbar^{-1}M] = \hbar^{-1}M, \quad
  [\bar{L},\hbar^{-1}\bar{M}] = \hbar^{-1}\bar{M}, 
\eeqnn
which are nothing but (\ref{hbar-LM-CCR}).  
Thus it is $\hbar^{-1}M$ and $\hbar^{-1}\bar{M}$ 
rather than $M$ and $\bar{M}$ that correspond to $\Delta$ 
and show up in generalized string equations.

\subsection{Classical limit of generalized string equations}

Let us turn to the case of double Hurwitz numbers. 
To derive a classical limit, we have to introduce $\hbar$ therein.  
At least formally, this can be done 
by rescaling the space-time variables as 
\beq
  s \to \hbar^{-1}s,\quad 
  \bst \to \hbar^{-1}\bst,\quad
  \bar{\bst} \to \hbar^{-1}\bar{\bst}. 
\label{rescaling-sttbar}
\eeq
The $\hbar$-independent Toda hierarchy is thereby 
converted to the $\hbar$-dependent form. Actually, 
it is rare that the rescaled Lax and Orlov-Schulman operators 
satisfy the condition (\ref{uv-cl-condition}).  
To achieve a meaningful (and nontrivial) classical limit, 
one should start from a carefully chosen 
{\it $\hbar$-dependent} solution 
of the {\it $\hbar$-independent} Toda hierarchy.  

In the language of the tau function \cite{TT93}, 
a meaningful classical limit is obtained 
from an $\hbar$-dependent tau function 
$\tau(\hbar,s,\bst,\bar{\bst})$ 
such that the rescaled tau function 
\beqnn
  \tau_\hbar(s,\bst,\bar{\bst}) 
  = \tau(\hbar,\hbar^{-1}s,\hbar^{-1}\bst,\hbar^{-1}\bar{\bst}) 
\eeqnn
behaves as 
\beq
  \log\tau_\hbar(s,\bst,\bar{\bst}) 
  = \hbar^{-2}\calF(s,\bst,\bar{\bst}) + O(\hbar^{-1}). 
\label{tau-qc-condition}
\eeq
$\calF = \calF(s,\bst,\bar{\bst})$ is called 
the ``free energy'' because of its relation 
to random matrices and topological field theories 
\cite{Dubrovin96,Krichever94}.  
If the rescaled tau function has this asymptotic form, 
the associated ``wave functions'' 
\beqnn
\begin{gathered}
  \Psi_\hbar(s,\bst,\bar{\bst},z) 
  = \frac{\tau_\hbar(s,\bst-\hbar[z^{-1}],\bar{\bst})}
         {\tau_\hbar(s,\bst,\bar{\bst})}
    z^{\hbar^{-1}s}e^{\hbar^{-1}\xi(\bst,z)},\\
  \bar{\Psi}_\hbar(s,\bst,\bar{\bst},z) 
  = \frac{\tau_\hbar(s+\hbar,\bst,\bar{\bst}-\hbar[z])}
         {\tau_\hbar(s,\bst,\bar{\bst})}
    z^{\hbar^{-1}s}e^{\hbar^{-1}\xi(\bar{\bst},z^{-1})},\\
  [z] = \left(z,z^2/2,\cdots,z^k/k,\cdots\right),\quad 
  \xi(\bst,z) = \sum_{k=1}^\infty t_kz^k 
\end{gathered}
\eeqnn
have the ``WKB'' form 
\beq
\begin{gathered}
  \Psi_\hbar(s,\bst,\bar{\bst},z) 
  = \exp\left(\hbar^{-1}S(s,\bst,\bar{\bst},z) + O(1)\right),\\
  \bar{\Psi}_\hbar(s,\bst,\bar{\bst},z) 
  = \exp\left(\hbar^{-1}\bar{S}(s,\bst,\bar{\bst},z) + O(1)\right) 
\end{gathered}
\eeq
and satisfy a set of auxiliary linear equations. 
The phase functions $S(s,\bst,\bar{\bst},z)$ 
and $\bar{S}(s,\bst,\bar{\bst},z)$ satisfy 
the associated Hamilton-Jacobi equations, 
which can be converted to the dispersionless Lax equations 
(\ref{dToda-Laxeq}) and the Poisson relations (\ref{dToda-CCR}) 
(see the review \cite{TT95} for details).  

An appropriate $\hbar$-dependent reformulation 
of the tau function (\ref{Hurwitz-tau}) 
of double Hurwitz numbers can be found 
by the following heuristic consideration.  
As we move into the $\hbar$-dependent formulation, 
the generalized string equations (\ref{str-eq}) are modified as 
\beqnn
    L = Qe^{-\beta/2}\bar{L} e^{\beta\hbar^{-1}\bar{M}},\quad 
  \bar{L}^{-1} = Qe^{\beta/2}L^{-1}e^{\beta\hbar^{-1}M}. 
\eeqnn
Obviously, these equations do not have a limit as $\hbar \to 0$.  
If, however, the parameter $\beta$ is simultaneously 
rescaled as 
\beq
  \beta \to \hbar\beta, 
\label{rescaling-beta}
\eeq
the generalized string equations are further modified as 
\beq
    L = Qe^{-\hbar\beta/2}\bar{L} e^{\beta\bar{M}},\quad 
  \bar{L}^{-1} = Qe^{\hbar\beta/2}L^{-1}e^{\beta M},  
\label{rescaled-str-eq}
\eeq
and have a meaningful classical limit of the form 
\beq
    \calL = Q\bar{\calL} e^{\beta\bar{\calM}},\quad 
  \bar{\calL}^{-1} = Q\calL^{-1}e^{\beta\calM}. 
\label{cl-str-eq}
\eeq

\subsection{Existence of $\hbar$-expansion}

To justify the foregoing heuristic derivation 
of the classical limit (\ref{cl-str-eq}) 
of the generalized string equations, 
let us show that the tau function (\ref{Hurwitz-tau}) 
with $\beta$ rescaled as (\ref{rescaling-beta}) 
does satisfy the condition (\ref{tau-qc-condition}).   

Recall the expression (\ref{Hurwitz-tau-Z}) of the tau function.  
After rescaling $s,\bst,\bar{\bst}$ and $\beta$ 
as (\ref{rescaling-sttbar}) and (\ref{rescaling-beta}), 
this expression is modified as 
\beq
  \tau_\hbar(s,\bst,\bar{\bst}) 
  = e^{\hbar^{-2}\beta s(s+\hbar)(2s+\hbar)/12}
    Q^{\hbar^{-2}s(s+\hbar)/2}
    Z_{\hbar\beta,e^{\beta(s+\hbar/2)}Q}[\hbar^{-1}\bst,\hbar^{-1}\bar{\bst}]. 
\label{hbar-Hurwitz-tau}
\eeq
Therefore it is sufficient to show that 
the logarithm of $Z_{\hbar\beta,Q}[\hbar^{-1}\bst,\hbar^{-1}\bar{\bst}]$ 
has an $\hbar$-expansion of the form 
\beq
  \log Z_{\hbar\beta,Q}[\hbar^{-1}\bst,\hbar^{-1}\bar{\bst}] 
  = \hbar^{-2}F_0 + F_1 + \hbar^2 F_2 + \cdots 
    + \hbar^{2n}F_n + \cdots, 
\label{logZ-expansion}
\eeq
where $F_0,F_1,\cdots$ are analytic functions 
of $(\beta,Q,\bst,\bar{\bst})$ in a common domain. 
The free energy is then given by 
\beq
  \calF = \frac{\beta s^3}{6} + \frac{s^2\log Q}{2} 
          + F_0(\beta,e^{\beta s}Q,\bst,\bar{\bst}). 
\eeq

(\ref{logZ-expansion}) is a generalization of the well known 
topological expansion for simple Hurwitz numbers. 
It is common in the literature that this kind of expansion 
is explained by a combination of topological and combinatorial 
consideration (see, e.g., Section 4.2 of Mironov and Morozov \cite{MM08}, 
Section 2.1 of Bouchard and Mari\~{n}o \cite{BM09} and 
Section 2.2 of Borot et al. \cite{BEMS09}). 
We take another approach based on the cut-and-join operator 
(\ref{boson-M_0}). 

\begin{theorem}
$\log Z_{\hbar\beta,Q}[\hbar^{-1}\bst,\hbar^{-1}\bar{\bst}]$ 
has an $\hbar$-expansion of the form (\ref{logZ-expansion}). 
\end{theorem}

\begin{proof}
Let $F = F(\hbar,\beta,Q,\bst,\bar{\bst})$ 
denote the left hand side of (\ref{logZ-expansion}) 
multiplied by $\hbar^2$.  By (\ref{Z[t,tbar]-M0}), 
$e^{\hbar^{-2}F}$ can be expressed as 
\beqnn
  e^{\hbar^{-2}F} = e^{\hbar\beta M_0(\hbar)}
      \exp\left(- \hbar^{-2}\sum_{k=1}^\infty Q^k kt_k\bar{t}_k\right), 
\eeqnn
where $M_0(\hbar)$ denotes the rescaled cut-and-join operator 
\beqnn
  M_0(\hbar) = \frac{1}{2}\sum_{j,k=1}^\infty 
      \left(\hbar^{-1}klt_kt_l\frac{\rd}{\rd t_{k+l}} 
         + \hbar(k+l)t_{k+l}\frac{\rd^2}{\rd t_k\rd t_l}\right). 
\eeqnn
Therefore $e^{\hbar^{-2}F}$ satisfies the differential equation 
\beqnn
  \frac{\rd e^{\hbar^{-2}F}}{\rd \beta} 
  = \hbar M_0(\hbar)e^{\hbar^{-2}F}
\eeqnn
with respect to $\beta$.  This equation can be 
further converted to the differential equation 
\beqnn
  \frac{\rd F}{\rd\beta} 
  = \frac{1}{2}\sum_{j,k=1}^\infty klt_kt_l\frac{\rd F}{\rd t_{k+l}} 
  + \frac{1}{2}\sum_{k,l=1}^\infty (k+l)t_{k+l}
       \left(\hbar^2\frac{\rd^2F}{\rd t_k\rd t_l} 
           + \frac{\rd F}{\rd t_k}\frac{\rd F}{\rd t_l} \right) 
\eeqnn
for $F$.  This equation is supplemented by the the initial condition 
\beqnn
  F|_{\beta=0} = - \sum_{k=1}^\infty Q^kkt_k\bar{t}_k. 
\eeqnn
We now seek a solution of this initial value problem 
in the form of a (formal) power series of $\hbar^2$: 
\beqnn
  F = F_0 + \hbar^2 F_1 + \cdots + \hbar^{2n}F_n + \cdots. 
\eeqnn
This reduces to solving the differential equations 
\begin{multline}
  \frac{\rd F_n}{\rd\beta} 
  = \frac{1}{2}\sum_{k,l=1}^\infty klt_kt_l\frac{\rd F_n}{\rd t_{k+l}} 
    + \frac{1}{2}\sum_{k,l=1}^\infty(k+l)t_{k+l} 
       \frac{\rd^2F_{n-1}}{\rd t_k\rd t_l} \\
  \mbox{} 
    + \frac{1}{2}\sum_{k,l=1}^\infty(k+l)t_{k+l}
        \sum_{m=0}^n\frac{\rd F_m}{\rd t_k}\frac{\rd F_{n-m}}{\rd t_l} 
\label{F_n-diffeq}
\end{multline}
for $n = 0,1,2,\cdots$ under the initial conditions 
\beq
  F_n|_{\beta=0} = - \delta_{n0}\sum_{k=1}^\infty Q^kkt_k\bar{t}_k. 
\label{F_n-beta=0}
\eeq
$F_n$'s are thereby recursively determined, and become 
analytic functions of $(\beta,Q,\bst,\bar{\bst})$ 
in a common domain of definition (because 
the differential equations other than the first one 
for $n = 0$ are linear with respect to $F_n$).  
Since the initial value problem for $F$ has a unique solution, 
the power series solution $F = F_0 + \hbar^2F_1 + \cdots$ 
should coincide with the left hand side 
of (\ref{logZ-expansion}).  
\end{proof}

One can thus confirm the expected asymptotic form 
(\ref{tau-qc-condition}) of the the rescaled tau function 
(\ref{hbar-Hurwitz-tau}).   Let us stress that 
this is also a proof for the case of simple Hurwitz numbers. 
To consider that case, one has only to set 
$\bar{t}_k = - \delta_{k1}$  
in the initial condition (\ref{F_n-beta=0}).  
Let us also point out that the main part $F_0$ 
of the free energy is determined by the $n = 0$ part 
of (\ref{F_n-diffeq}) and (\ref{F_n-beta=0}), namely, 
\beq
  \frac{\rd F_0}{\rd\beta} 
  = \frac{1}{2}\sum_{k,l=1}^\infty klt_kt_l\frac{\rd F_0}{\rd t_{k+l}} 
    + \frac{1}{2}\sum_{k,l=1}^\infty(k+l)t_{k+l}
        \frac{\rd F_0}{\rd t_k}\frac{\rd F_0}{\rd t_l} 
\eeq
and 
\beq
  F_0|_{\beta=0} = - \sum_{k=1}^\infty Q^kkt_k\bar{t}_k. 
\eeq
It will be interesting to apply the diagramatic technique 
of Mironov and Morozov \cite{MM08} to these equations.

\section{Solution of classical limit of generalized string equations}

\subsection{Comparison with generalized string equations of $c = 1$ string theory}

Our goal in this section is to solve 
the generalized string equations (\ref{cl-str-eq}) 
and to derive some implications thereof. 
To this end, it is instructive to compare these equations 
with the generalized string equations of $c = 1$ string theory
\cite{DMP93,HOP94,EK94,Takasaki95}.  

In the classical limit, the generalized string equations 
of $c = 1$ string theory read
\beq
  \calL = \bar{\calL}\bar{\calM},\quad 
  \bar{\calL}^{-1} = \calL^{-1}\calM. 
\label{c=1-str-eq}
\eeq
Let us mention that the same equations play a central role 
in a problem of complex analysis and its applications 
to interface dynamics \cite{WZ00,MWWZ00,Zabrodin01}.  
Actually, it is the equation 
\beq
  \{\calL,\bar{\calL}^{-1}\} = 1 
\label{c=1-str-eq2}
\eeq
rather than (\ref{c=1-str-eq}) that is referred to 
as a ``string equation'' in these applications. 
Note that one can readily derive (\ref{c=1-str-eq2}) 
from (\ref{c=1-str-eq}).   In the same sense, 
one can derive the equation 
\beq
  \{\log\calL,\log\bar{\calL}^{-1}\} = \beta 
\label{cl-str-eq2}
\eeq
from (\ref{cl-str-eq}) as a counterpart of (\ref{c=1-str-eq}) 
for double Hurwitz numbers.  

(\ref{c=1-str-eq2}) and (\ref{cl-str-eq2}) resemble 
the string equation (or the Douglas equation) 
\beq
  [Q,P] = 1 
\label{Douglas-eq}
\eeq
and its classical limit 
\beq
  \{Q,P\} = 1
\label{cl-Douglas-eq}
\eeq
in two-dimensional quantum gravity \cite{Douglas90,Schwarz91,KS91}. 
$Q$ and $P$ in (\ref{Douglas-eq}) are one-dimensional 
differential operators of the form 
\beqnn
  Q = \rd_x^n + u_2\rd_x^{n-2} + \cdots + u_n,\quad
  P = \rd_x^m + v_2\rd_x^{m-2} + \cdots + v_n. 
\eeqnn
In the classical limit, $\rd_x$ is replaced by 
a variable $p$ with the Poisson bracket 
\beqnn
  \{p,x\} = 1, 
\eeqnn
and $Q$ and $P$ are polynomials of the form 
\beqnn
  \calQ = p^n + u_2p^{n-2} + \cdots + u_n,\quad
  \calP = p^m + v_2p^{m-2} + \cdots + v_n. 
\eeqnn
In this setting, $\calQ$ and $\calP$ may be thought of 
as coordinates of the spectral curve (parametrized by $p$) 
in the sense of Eynard and Orantin \cite{EO07}.  
Namely, when $x$ and other deformation variables 
(time variables of the underlying KP hierarchy) 
are fixed to special values, $\calQ$ and $\calP$ satisfy 
a defining equation $f(X,Y) = 0$ of the spectral curve as 
\beqnn
  f(\calQ,\calP) = 0. 
\eeqnn

Although the KP and Toda hierarchies are different in nature, 
the last remark seems to suggest that one may think of 
an equation of the form 
\beqnn
  f(\calL,\bar{\calL}^{-1}) = 0 
\eeqnn
as the spectral curve in the present setting.  
This observation is partly supported by the fact 
that such a curve is derived as the spectral curve 
in the random matrix approach to $c = 1$ string theory 
and interface dynamics \cite{AKK03,TBAZW04a,TBAZW04b}. 

Bearing these remarks in mind, let us turn to the issue 
of solving the generalized string equations (\ref{cl-str-eq}). 
These equations, like (\ref{c=1-str-eq}), 
are a kind of ``nonlinear Riemann-Hilbert problems''.  
One can use the method developed for solving 
(\ref{c=1-str-eq}) \cite{Takasaki95} 
to construct a solution of (\ref{cl-str-eq}) 
as power series of $\bst$ and $\bar{\bst}$. 
As it turns out, (\ref{cl-str-eq}) can be treated 
in much the same way apart from technical complications.

\subsection{Decomposition of equations}

Let us convert (\ref{cl-str-eq}) to the logarithmic form 
\beq
\begin{gathered}
  \log(\calL p^{-1}) 
    = \log Q - \log(\bar{\calL}^{-1}p) + \beta\bar{\calM},\\
  \log(\bar{\calL}^{-1}p) 
    = \log Q - \log(\calL p^{-1}) + \beta\calM. 
\end{gathered}
\label{log-str-eq}
\eeq
Since $\calL p^{-1}$ and $\bar{\calL}^{-1}p$ 
are Laurent series of the form 
\beqnn
  \calL p^{-1} = 1 + \sum_{n=1}^\infty u_np^{-n},\quad 
  \bar{\calL}^{-1}p = \bar{u}_0 + \sum_{n=1}^\infty \bar{u}_np^n 
\eeqnn
with nonzero leading terms, one can expand the logarithm as 
\beqnn
  \log(\calL p^{-1}) = \sum_{n=1}^\infty \alpha_np^{-n},\quad 
  \log(\bar{\calL}^{-1}p) 
    = \log\bar{u}_0 + \sum_{n=1}^\infty \bar{\alpha}_np^n, 
\eeqnn
where 
\beqnn
\begin{gathered}
  \alpha_n = u_n + \mbox{(polynomial of $u_1,\cdots,u_{n-1}$)},\\
  \bar{\alpha}_n 
    = \bar{u}_0^{-1}\bar{u}_n
      + \mbox{(polynomial of 
        $\bar{u}_0^{-1}\bar{u}_1,\cdots,\bar{u}_0^{-1}\bar{u}_{n-1}$)}. 
\end{gathered}
\eeqnn

We now substitute 
\beqnn
\begin{gathered}
  \bar{\calM} = - \sum_{k=1}^\infty k\bar{t}_k\bar{\calL}^{-k} 
          + s + \sum_{n=1}^\infty\bar{v}_n\bar{\calL}^n,\\
  \calM = \sum_{k=1}^\infty kt_k\calL^k 
          + s + \sum_{n=1}^\infty v_n\calL^{-n}
\end{gathered}
\eeqnn
in (\ref{log-str-eq}) and expand both hand sides in powers of $p$. 
This leads to an infinite set of equations for the coefficients 
$u_n,\bar{u}_n,v_n,\bar{v}_n$ of $\calL,\calM,\bar{\calL},\bar{\calM}$ 
as follows.    

Equating the coefficients of $p^{-n}$, $n = 0,1,\cdots$, 
in both hand sides of the first equation of (\ref{log-str-eq}) 
gives the equations 
\begin{gather}
  0 = \log Q - \log\bar{u}_0 + \beta s 
      - \beta\sum_{k=1}^\infty k\bar{t}_k(\bar{\calL}^{-k})_0, 
  \label{log-str-eq(0bar)}\\
  \alpha_n = - \beta\sum_{k=n}^\infty k\bar{t}_k(\bar{\calL}^{-k})_{-n}, 
        \quad n = 1,2,\cdots,
  \label{log-str-eq(nbar)} 
\end{gather}
where $(\bar{\calL}^{-k})_{-n}$ stands for 
the coefficient of $p^{-n}$ in $\bar{\calL}^{-k}$.  
In the same way, the coefficients of $p^n$, $n = 0,1,2,\cdots$, 
in the second equation of (\ref{log-str-eq}) give the equations 
\begin{align}
  \log\bar{u}_0 &= \log Q + \beta s 
      + \beta\sum_{k=1}^\infty kt_k(\calL^k)_0, 
  \label{log-str-eq(0)}\\
  \bar{\alpha}_n &= \beta\sum_{k=n}^\infty kt_k(\calL^k)_n, 
               \quad n = 1,2,\cdots, 
  \label{log-str-eq(n)}
\end{align}
where $(\calL^k)_n$ denotes the coefficient of $p^n$ in $\calL^k$.  
Since 
\beqnn
  (\calL^k)_n = (\bar{\calL}^{-k})_{-n} = 0 \quad \mbox{for $k<n$}, 
\eeqnn
the range of $k$ in the sums of (\ref{log-str-eq(nbar)}) 
and (\ref{log-str-eq(n)}) is limited to $k \ge n$.  
Note that one can use the formal residue notation 
\beqnn
  \res\left(\sum_n a_np^ndp\right) = a_{-1} 
\eeqnn
to express $\left(\calL^k\right)_n$ and 
$\left(\bar{\calL}^{-k}\right)_{-n}$ as 
\beqnn
  (\calL^k)_n 
    = \res\left(\calL^kp^{-n}d\log p\right),\quad 
  (\bar{\calL}^{-k})_{-n} 
    = \res\left(\bar{\calL}^{-k}p^nd\log p\right). 
\eeqnn
Such an expression turns out to be useful in the subsequent 
consideration.   Since $v_n$'s and $\bar{v}_n$'s are absent, 
(\ref{log-str-eq(0bar)}) -- (\ref{log-str-eq(n)}) 
are equations for $u_n$'s and $\bar{u}_b$'s only.  

The remaining part of (\ref{log-str-eq}) determine 
$v_n$'s and $\bar{v}_n$'s.  To see this, 
it is more convenient to expand (\ref{log-str-eq}) 
in powers of $\calL$ and $\bar{\calL}$ rather than of $p$.  
Extracting the coefficients of $\bar{\calL}^n$ 
from the first equation of (\ref{cl-str-eq}) and 
those of $\calL^{-n}$ from the second equation 
yields the equations 
\begin{align}
  0 &= - \res\left(\log(\bar{\calL}^{-1}p)
         \bar{\calL}^{-n}d\log\bar{\calL}\right)
       + \beta\bar{v}_n, 
  \label{log-str-eq(vbar_n)}\\
  0 &= - \res\left(\log(\calL p^{-1})\calL^nd\log\calL\right) 
       + \beta v_n 
  \label{log-str-eq(v_n)}
\end{align}
for  $n = 1,2,\cdots$.  Thus $v_n$'s and $\bar{v}_n$'s 
are determined by $\calL$ and $\bar{\calL}$ as 
\beqnn
\begin{gathered}
  v_n = \beta^{-1}\res\left(\log(\calL p^{-1})\calL^nd\log\calL\right),\\
  \bar{v}_n = \beta^{-1}\res\left(\log(\bar{\calL}^{-1}p)
              \bar{\calL}^{-n}d\log\bar{\calL}\right). 
\end{gathered}
\eeqnn

(\ref{log-str-eq}) can be thus decomposed into the infinite set 
of equations (\ref{log-str-eq(0bar)}) -- (\ref{log-str-eq(n)}), 
(\ref{log-str-eq(vbar_n)}) and (\ref{log-str-eq(v_n)}).   
The next task is to show that they do have a solution.

\subsection{Solution of equations}

As we observed above, one can think of 
(\ref{log-str-eq(0bar)}) -- (\ref{log-str-eq(n)}) 
as equations for $u_n$'s and $\bar{u}_n$'s.   
We want to construct these functions as power series 
of $(\bst$ and $\bar{\bst}$ with coefficients depending on $s$. 
If $u_k$'s and $\bar{u}_k$ are thus constructed, 
$v_k$ and $\bar{v}_k$'s are determined by 
(\ref{log-str-eq(vbar_n)}) and (\ref{log-str-eq(v_n)}).  

Let us first examine (\ref{log-str-eq(0bar)}) 
and (\ref{log-str-eq(0)}).  Since 
\beqnn
\begin{gathered}
  \alpha_n = u_n  
    + \mbox{(monomials of higher degrees in $u_1,\cdots,u_{n-1}$)},\\
  \bar{\alpha}_n = \bar{u}_0^{-1}\bar{u}_n 
    + \mbox{(monomials of higher degrees in 
       $\bar{u}_0^{-1}\bar{u}_,\cdots,\bar{u}_0^{-1}\bar{u}_{n-1}$)}, 
\end{gathered}
\eeqnn
$u_n$ and $\bar{u}_n$, $n = 1,2,\cdots$, show up 
on the left hand side of these equations linearly.  
The right hand side consists of terms 
that are multiplied by $t_k$'s and $\bar{t}_k$'s.  
Therefore these equations yield a huge system of 
recursion relations for the coefficients of 
power series expansion of $u_n$'s and $\bar{u}_n$'s. 

Note that $\bar{u}_0$, which remains to be determined, 
is contained in the coefficients of the power series 
expansion of $u_n$ and $\bar{u}_n$.  
We need another equation to determine $\bar{u}_0$.  
Actually, there are two equations 
(\ref{log-str-eq(0bar)}) and (\ref{log-str-eq(0)}) 
rather than just one.   This puzzle is resolved as follows 
\footnote{A similar result holds for the generalized 
string equations (\ref{c=1-str-eq}) of $c = 1$ string theory.  
This fills a logical gap left in our previous paper 
\cite{Takasaki95}.}. 

\begin{lemma}
If (\ref{log-str-eq(nbar)}) and (\ref{log-str-eq(n)}) 
are satisfied, then (\ref{log-str-eq(0bar)}) and 
(\ref{log-str-eq(0)}) are equivalent, and reduces 
to the equation
\beq
  \log\bar{u}_0 = \log Q + \beta s 
     + \sum_{k=1}^\infty k\alpha_k\bar{\alpha}_k. 
\label{log-str-eq(0reduced)}
\eeq
\end{lemma}

\begin{proof}
Let us use the formal residue notation to express 
the terms $(\calL^k)_0$ in (\ref{log-str-eq(0bar)}) as 
\beqnn
  (\calL^k)_0 = \res(\calL^kd\log p).
\eeqnn
Since the identity 
\beqnn
  0 = \res(\calL^kd\log\calL) 
    = \res\left(\calL^kp\frac{\rd\log\calL}{\rd p}d\log p\right) 
\eeqnn
holds for $k \ge 1$, this expression of $(\calL^k)_0$ 
can be further rewritten as 
\beqnn
\begin{aligned}
  (\calL^k)_0 
  &= \res\left(\calL^k(1 - p\frac{\rd\log\calL}{\rd p})d\log p\right)\\
  &= - \res\left(\calL^k\frac{\rd\log(\calL p^{-1})}{\rd p}dp\right). 
\end{aligned}
\eeqnn
By substituting $\calL p^{-1} = \alpha_1p^{-1} + \alpha_2 p^{-2} + \cdots$, 
the right hand side can be expanded as 
\beqnn
\begin{aligned}
  \mbox{RHS} 
  &= \alpha_1\res(\calL^k p^{-2}dp) 
     + 2\alpha_2\res(\calL^k p^{-3}dp) + \cdots \\
  &= \alpha_1(\calL^k)_1 + 2\alpha_2(\calL^k)_2 + \cdots. 
\end{aligned}
\eeqnn
Since $(\calL^k)_n = 0$ for $n > k$, this expansion terminates 
at the $k$-th term.  One can thus obtain the identity 
\beqnn
  (\calL^k)_0 
  = \alpha_1(\calL^k)_1 + 2\alpha_2(\calL^k)_2 
     + \cdots + k\alpha_k(\calL^k)_k. 
\eeqnn
In the same way, one can derive the identity 
\beqnn
  (\bar{\calL}^{-k})_0 
  = \bar{\alpha}_1(\bar{\calL}^{-k})_{-1} 
    + 2\bar{\alpha}_2(\bar{\calL}^{-k})_{-2} 
    + \cdots + k\bar{\alpha}_k(\bar{\calL}^{-k})_{-k}. 
\eeqnn
By virtue of these identities, one can rewrite the two sums 
in (\ref{log-str-eq(0bar)}) and (\ref{log-str-eq(0)}) as 
\beqnn
\begin{aligned}
  \sum_{k=1}^\infty kt_k(\calL^k)_0 
  &= \sum_{k=1}^\infty kt_k\left( 
      \alpha_1(\calL^k)_1 + 2\alpha_2(\calL^k)_2 
      + \cdots + k\alpha_k(\calL^k)_k \right) \nonumber\\
  &= \alpha_1\sum_{k=1}^\infty kt_k(\calL^k)_1 
     + 2\alpha_2\sum_{k=1}^\infty kt_k(\calL^k)_2 + \cdots \nonumber\\
  &= \beta^{-1}\sum_{n=1}^\infty n\alpha_n\bar{\alpha}_n 
\end{aligned}
\eeqnn
and 
\beqnn
\begin{aligned}
  \sum_{k=1}^\infty k\bar{t}_k(\bar{\calL}^{-k})_0 
  &= \sum_{k=1}^\infty k\bar{t}_k\left( 
      \bar{\alpha}_1(\bar{\calL}^{-k})_1 
      + 2\bar{\alpha}_2(\bar{\calL}^{-k})_2 
      + \cdots + k\bar{\alpha}_k(\bar{\calL}^{-k})_k \right) \nonumber\\
  &= \bar{\alpha}_1\sum_{k=1}^\infty k\bar{t}_k(\bar{\calL}^{-k})_1 
     + 2\bar{\alpha}_2\sum_{k=1}^\infty k\bar{t}_k(\bar{\calL}^{-k})_2 
     + \cdots \nonumber\\
  &= - \beta^{-1}\sum_{n=1}^\infty n\bar{\alpha}_n\alpha_n. 
\end{aligned}
\eeqnn
Note that (\ref{log-str-eq(n)}) and (\ref{log-str-eq(nbar)}) 
have been used to derive the last lines. 
Thus (\ref{log-str-eq(0bar)}) and (\ref{log-str-eq(0)}) 
turn out to reduce to the same equation (\ref{log-str-eq(0reduced)}). 
\end{proof}

We can thus use (\ref{log-str-eq(0reduced)}) 
in place of (\ref{log-str-eq(0bar)}) and 
(\ref{log-str-eq(0)}).  Adding this equation 
to (\ref{log-str-eq(nbar)}) and (\ref{log-str-eq(n)}), 
we obtain a full system of equations that determine 
the power series expansion of $u_n$, $\bar{u}_n$ and 
$\bar{u}_0$ recursively.   We can readily see 
from (\ref{log-str-eq(0reduced)}) that $\bar{u}_0$ 
is a power series of the form 
\beq
  \log\bar{u}_0 = \log Q + \beta s 
    + \mbox{(terms of positive orders in $\bst,\bar{\bst}$)}. 
\label{u_0-expansion}
\eeq
On the other hand, since 
\beqnn
\begin{gathered}
  (\calL^n)_n = 1,\quad
  (\bar{\calL}^{-n})_{-n} = \bar{u}_0^n,\\
  (\calL^k)_n = \mbox{(polynomial in $u_1,\cdots,u_{k-n}$)}
    \quad \mbox{for $k>n$},\\
  (\bar{\calL}^{-k})_{-n} = \bar{u}_0^k \times\mbox{(polynomial in 
       $\bar{u}_0^{-1}\bar{u}_1,\cdots,\bar{u}_0^{-1}\bar{u}_{k-n}$)}
    \quad \mbox{for $k>n$}, 
\end{gathered}
\eeqnn
$u_n$ and $\bar{u}_n$ are power series of the form 
\beq
\begin{gathered}
  u_n = - \beta n\bar{t}_n\bar{u}_0^n 
         + \mbox{(terms of higher orders in $\bst,\bar{\bst}$)},\\
  \bar{u}_n = \beta nt_n\bar{u}_0  
         + \mbox{(terms of higher orders in $\bst,\bar{\bst}$)}. 
\end{gathered}
\label{uubar-expansion}
\eeq

This power series solution of (\ref{cl-str-eq}) 
(which is unique by construction) is homogeneous 
just like the solution of the generalized string equations 
(\ref{c=1-str-eq}) for $c = 1$ string theory \cite{Takasaki95}.  
This is a consequence of invariance of the the string equations 
under the scaling transformations 
\beq
\begin{gathered}
  t_n \to c^{-n}t_n, \quad \bar{t}_n \to c^n\bar{t}_n, \quad
  s \to s, \quad p \to cp, \\
  u_n \to c^nu_n, \quad \bar{u}_n \to c^{-n}\bar{u}_n, \quad 
  v_n \to c^nv_n, \quad \bar{v}_n \to c^{-n}\bar{v}_n. 
\end{gathered}
\label{solution-homogeneity}
\eeq

In summary, we have observed the following: 

\begin{theorem}
The generalized string equations (\ref{cl-str-eq}) 
have a unique solution that has power series expansion 
with respect to $(\bst,\bar{\bst})$ as shown 
in (\ref{u_0-expansion}) and (\ref{uubar-expansion}). 
This solution is homogeneous with respect to 
the scaling transformation (\ref{solution-homogeneity}). 
\end{theorem}

\subsection{Solutions at special values of $\bst,\bar{\bst}$}

The foregoing construction of solution simplifies to some extent 
when $\bst$ and $\bar{\bst}$ take special values.  
Of particular interest are the following two cases: 
\begin{itemize}
\item[(i)] $t_k$'s are free, and $\bar{t}_k$'s are restricted 
to $\bar{t}_k = \bar{t}_1\delta_{k1}$, 
\item[(ii)] $\bar{t}_k$'s are restricted to 
$\bar{t}_k = \bar{t}_1\delta_{k1}$, and $t_k$'s are free. 
\end{itemize}
They amount to restricting the generating function 
$Z[\bst,\bar{\bst}]$ of double Hurwitz numbers 
to generating functions of simple Hurwitz numbers.  
Since these two cases are essentially equivalent, 
let us consider (i) only.   

In the case of (i), (\ref{log-str-eq(nbar)}) implies 
that $\alpha_n$ vanishes for $n > 1$ and 
that the only non-vanishing component is given by 
\beqnn
  \alpha_1 = u_1 = - \beta\bar{t}_1\bar{u}_0. 
\eeqnn
Thus $\calL$ simplifies as 
\beq
  \calL = pe^{\alpha_1p^{-1}} = pe^{-\beta\bar{t}_1\bar{u}_0p^{-1}}, 
\label{L-case(i)}
\eeq
and $(\calL^k)_n$ can be written explicitly as 
\beq
  (\calL^k)_n = \frac{(ku_1)^{k-n}}{(k-n)!} 
    = \frac{(-k\beta\bar{t}_1\bar{u}_0)^{k-n}}{(k-n)!}. 
\label{(L^k)_n-case(i)}
\eeq
$\bar{u}_n$, $n = 1,2,\cdots$, are thereby 
recursively determined by (\ref{log-str-eq(n)}) 
as a function of $t_k$'s and $\bar{u}_0$.  
$\bar{u}_0$ is determined by (\ref{log-str-eq(0)}), 
which now takes an explicit form as 
\beq
  \log\bar{u}_0 
  = \log Q + \beta s 
    + \beta\sum_{k=1}^\infty 
         kt_k\frac{(-k\beta\bar{t}_1\bar{u}_0)^k}{k!}. 
\label{u_0-case(i)}
\eeq
$v_n$ and $\bar{v}_n$'s, too, have more or less 
explicit formulae, though we omit details. 

This result shows a remarkable feature.  
Namely, (\ref{L-case(i)}) resembles the defining equation 
$x = ye^y$ of Lambert's W-function $y = W(x)$ 
if $\calL^{-1}$ and $p^{-1}$ are identified with $x$ and $y$.   
The so called Lambert curve is defined by this equation 
on the $(x,y)$-plane, and plays a fundamental role 
in the recent studies on Hurwitz numbers 
\cite{BM09,BEMS09,EMS09,MS09a,MS09b}.  

This analogy becomes more precise when $t_k$'s, 
too, are specialized to $t_k = 0$, $k = 1,2,\cdots$.  
In that case, $\bar{u}_0$ is explicitly determined as 
\beqnn
  \bar{u}_0 = Qe^{\beta s}. 
\eeqnn
Moreover, (\ref{log-str-eq(n)}) implies that $\bar{\alpha}_n$ 
vanishes for all $n$, hence 
\beq
  \bar{\calL}^{-1} = \bar{u}_0p^{-1}. 
\eeq
(\ref{L-case(i)}) thereby turns into the equation 
\beq
  \calL = \bar{u}_0\bar{\calL}e^{-\beta\bar{t}_1\bar{\calL}^{-1}} 
\label{LLbar-eq-case(i)}
\eeq
for $\calL$ and $\bar{\calL}$.  In view of the remarks 
in the beginning of this section, it seems likely 
that this equation can be identified with 
the spectral curve for simple Hurwitz numbers.

\subsection*{Acknowledgements}

This work is partly supported by JSPS Grants-in-Aid 
for Scientific Research No. 19104002, No. 21540218 
and No. 22540186 from the Japan Society 
for the Promotion of Science.

\end{document}